\documentclass{article}
\usepackage[T1]{fontenc}

\usepackage[usenames,dvipsnames]{xcolor}
\usepackage{xspace}
\usepackage{bbold}
\usepackage{graphicx}
\usepackage{algpseudocode}
\usepackage{algorithm}
\usepackage{authblk}
\usepackage{amsthm}

\usepackage{hyperref}
\hypersetup{
  colorlinks,
  citecolor=black,
  filecolor=black,
  linkcolor=blue,
  urlcolor=blue
}

\usepackage{fancyhdr}
\usepackage{xcolor}
\usepackage{amsmath,amssymb,mathtext}

\usepackage[normalem]{ulem}

\newtheorem{hypothesis}{Hypothesis}

\newtheorem{theorem}{Theorem}

\newtheorem{definition}{Definition}
\newtheorem{remark}{Remark}
\newtheorem{lemma}{Lemma}
\newtheorem{claim}{Claim}
\newtheorem{proposition}{Proposition}
\newtheorem{corollary}{Corollary}

\newcommand\vphi{\varphi}

\renewcommand\le{\leqslant}

\renewcommand\ge{\geqslant}

\newcommand\proph{\textbf{PROPH}}
\newcommand\eproph{\textbf{EPROPH}}
\DeclareMathOperator\dist{dist}

\DeclareMathOperator\cl{cl}
\DeclareMathOperator\Span{span}

\begin{document}
\title{Non-Adaptive Prophet Inequalities for Minor-Closed Classes of Matroids}
\author{Kanstantsin Pashkovich, Alice Sayutina}
\affil{University of Waterloo\\
Department of Combinatorics \& Optimization\\
200 University Avenue West\\
Waterloo, ON, Canada\\
N2L 3G1}
\date{}\maketitle

\begin{abstract}

  We consider the matroid prophet inequality problem. This problem has been extensively studied in the case of adaptive mechanisms. In particular, 
there is a tight $2$-competitive mechanism for all matroids~\cite{kw2012matroid}.

  However, it is not known what classes of matroids admit non-adaptive mechanisms with constant guarantee. Recently, in~\cite{chawla2020non} it was shown that there
  are constant-competitive non-adaptive mechanisms for graphic
  matroids. In this work, we show that various known classes of matroids admit  constant-competitive non-adaptive mechanisms. 
\end{abstract}

\section{Introduction}
Let us consider the \emph{classical  prophet inequality problem}~\cite{krengel1977semiamarts}. A \emph{gambler} observes a sequence of non-negative independent random variables $X_1$, $X_2$, \ldots, $X_n$, which correspond to a sequence of values for $n$ \emph{items}. The gambler knows the  distributions of $X_1$, $X_2$, \ldots, $X_n$.  The gambler is allowed to accept at most one item; and the gambler is interested in maximizing the value of the accepted item. However, the gambler cannot simply select an item of the maximum value, because the values of the $n$ items are revealed to the gambler one by one; and each time a value of the current item is revealed the gambler has to make an irrevocable choice whether to accept the current item or not.

What stopping rule the gambler should use to maximize the expected
value of the item they accept? The gambler knows only the distributions of $X_1$, $X_2$, \ldots, $X_n$ while a \emph{prophet} knows the realization of $X_1$, $X_2$, \ldots, $X_n$. Thus, in contrast to the gambler the prophet can always obtain the maximum item's value. The seminal result of Krengel
and Sucheston~\cite{krengel1977semiamarts} showed that the gambler can obtain at least a half of the expected value obtained by
the prophet.

The classical prophet inequality problem led to a series of works on different variants of the problem. A natural variant of the problem is the generalization of the problem where a gambler can buy more than one item, but the set of bought items should satisfy a known feasibility constraint.
Formally, let us be given a collection~$\mathcal{S}\subseteq 2^{[n]}$ of item sets. Then both
gambler and prophet can select any item set~$S$ from~$\mathcal{S}$. So $\mathcal S$ defines a feasibility constraint for selecting items. In most standard examples of feasibility  constraints,
$\mathcal{S}$ can be defined as a collection of all item sets with cardinality at most $k$ for some natural number~$k$. More generally $\mathcal{S}$ can be defined as a collection
of all independent sets in some matroid, in this case we speak about the~\emph{matroid prophet inequality problem}.

The result in \cite{samuelcahn1984comparison} showed that in the single-item setting
a gambler can obtain at least half of the prophet value by using the following threshold-rule:
determine a constant $T$ as a function of known distributions and accept the first
item exceeding $T$. This rule results in a $2$-competitive mechanism, similar to the
adaptive approach of \cite{krengel1977semiamarts}. Note, that this approximation guarantee is known to be tight.
There is also another method to set a threshold presented in \cite{kw2012matroid},
which also results in a $2$-competitive mechanism. This was extended by Chawla et al.
in \cite{chawla2010multi} and \cite{chawla2020non} to the setting of several items.

The results presented in \cite{kw2012matroid} further extend to the \emph{matroid prophet
inequalities}, where accepted items need to form an independent set in a known matroid. It leads to
a $2$-competitive mechanism for every matroid, matching the single-item setting result. However,
unlike the mechanism in the single-item setting, the mechanism for matroids is adaptive:
the thresholds for  items are computed based on the
previously accepted items. By \cite{kw2012matroid}, there also exists a constant-competitive
adaptive mechanism for feasibility constraints defined as an intersection of constant
number of matroids. The mechanism by Kleinberg and Weinberg was further extended to a $2$-competitive mechanism for
polymatroids by D\"utting and Kleinberg in \cite{dutting2015polymatroid}.

Gravin and Wang \cite{gravin2019prophet} studied the bipartite matching version of this problem: in their version, the arriving items are the edges of the (known) bipartite graph. Gravin and Wang obtained a $3$-competitive non-adaptive mechanism, which assigns thresholds to each vertex in the graph and an edge is accepted only if its weight is at least the sum of  the thresholds associated with its endpoints.

Feldman, Svensson and Zenklusen \cite{feldman2016online} studied online item selection mechanisms called ``online contention resolution schemes" (OCRS). They showed that given special properties, OCRS translate directly into a constant-competitive prophet inequality for the same problem against almighty adversary, i.e. an adversary which knows in advance realizations of all the items and the random bits generated by an algorithm. As a result, they develop a constant-competitive mechanism for prophet inequalities of the intersection of a constant number of matroids, knapsack and matching constraints. Those mechanisms are ``almost'' non-adaptive in a sense that they fix thresholds for all items, however their mechanisms also impose a subconstraint: an item cannot be accepted if together with previously accepted items it forms one of the fixed forbidden sets.

Finally, in a later version of their paper \cite{feldman2021online}, they prove that pure non-adaptive mechanisms cannot achieve a constant-competitive approximation even against a ``normal'' adversary. They construct a family of gammoid matroids showing a lower bound of $\Omega(\log n{/}\log\log n)$ for a guarantee of non-adaptive mechanisms on gammoids with $n$ elements.

There have been works studying similar setups with other goals. Chawla et al. \cite{chawla2010multi} studied a Bayesian item selection process in a fixed item arrival order or against an adversary in control of the order. They studied it from a perspective of the revenue maximization for the auctioneer. The performance is constant-competitive
compared to the well-known Myerson  mechanism~ \cite{myerson1981optimal}, which achieves the largest possible expected revenue among truthful mechanisms.
The mechanism by Chawla et al. \cite{chawla2010multi}  has an advantage that it determines static thresholds together with a subconstraint so that each agent can be offered take-it-or-leave-it prices in an online fashion.

Recently, Chawla et al. \cite{chawla2020non} developed a $32$-competitive non-adaptive mechanism for graphic matroids against 
adversary item ordering.

\subsection{Our results}

First, we list the known results for non-adaptive mechanism that were mentioned in the previous section.

\begin{theorem}[\textbf{Uniform Rank $1$ Matroid}~\cite{samuelcahn1984comparison}] 
  \label{thm:twocompetitive}
  There exists a $2$-competitive non-adaptive mechanism for single-item setting.
\end{theorem}

\begin{theorem}[\textbf{Graphic Matroid}~\cite{chawla2020non}]\label{thm:graphic32} There exists a $32$-competitive non-adaptive mechanism for graphic matroids.
\end{theorem}

Now let us list our results. In case of a simple graph, i.e. a graph with no parallel edges or loops, we can slightly improve the above theorem by considering essentially the same mechanism as~\cite{chawla2020non} but considering a different scaling of a point from the matroid polytope. We provide this result for the sake of completeness.

\begin{theorem}\label{thm:graphic16} There exists a $16$-competitive non-adaptive mechanism for graphic matroids  in the case of simple graphs.
\end{theorem}

Furthermore, the mechanism~\cite{chawla2020non} can be generalized to the setting of $k$-column sparse matroids. This result we need later to obtain Theorem~\ref{thm:minor_closed}.

\begin{theorem}[\textbf{$k$-Column Sparse Matroids}] \label{thm:k_column_sparse} There exists a $(2^{k+2} k)$-competitive non-adaptive mechanism for $k$-column sparse matroids.
\end{theorem}

Note, that Theorem~\ref{thm:graphic32} of~\cite{chawla2020non} follows from Theorem~\ref{thm:k_column_sparse}, since a graphic matroid  is also a $2$-column sparse matroid over $\mathbb{F}_2$.

Using analogous approach to the one in \cite{soto2013matroid}, we also develop a mechanism for cographic matroids.

\begin{theorem}[\textbf{Cographic Matroids}] \label{thm:cographic} There exists a $6$-competitive non-adaptive mechanism for cographic matroids.
\end{theorem}

The approach in~\cite{soto2013matroid} immediately leads to the following result for $\gamma$-sparse matroids.
\begin{theorem}[\textbf{$\gamma$-Sparse Matroids}~\cite{soto2013matroid}]  There exists a $\gamma$-competitive non-adaptive mechanism for $\gamma$-sparse matroids.
\end{theorem}

Combining the above results and using classic Seymour's decomposition results we obtain the following theorem.

\begin{theorem}[\textbf{Regular Matroids}]~\label{thm:regularcomp} There exists a $256$-competitive non-adaptive mechanism for regular matroids.
\end{theorem}

Subject to the Structural Hypothesis~\ref{hypothesis:tree_decomp} due to Geelen, Gerards and Whittle, which is stated later, we can also derive the following result.

\begin{theorem}\label{thm:minor_closed} Subject to the Structural Hypothesis~\ref{hypothesis:tree_decomp}, for every prime number~$p$ there exists
  a constant-competitive mechanism for every proper minor-closed class of matroids representable over $\mathbb{F}_p$.
\end{theorem}

We also would like to observe that some of the recent results on ``single sample prophet inequalities'' (SSPI) lead to non-adaptive constant-competitive mechanisms. For this, the single sample required by the gambler in SSPI can be directly sampled by our gambler from the available distributions. In particular, the results in~\cite{azar} and~\cite{caramanis} on laminar matroids and truncated partition matroids inspired by the mechanism in~\cite{ma} lead to non-adaptive mechanisms for prophet inequalities. To obtain these results, it is crucial that the mechanism in~\cite{ma} does not involve subconstraints, i.e. each item is accepted as long as the item is not in the ``observation phase'', the item passes its threshold based only on the ``observation phase'' and the item forms an independent set with previously accepted items. In comparison, it is not clear how from the results on regular matroids in~\cite{azar} based on the mechanism in~\cite{dinitz} one can obtain non-adaptive mechanisms.

So the following results can be directly obtained from~\cite{azar} and~\cite{caramanis}, respectively.

\begin{theorem}[\textbf{Laminar Matroid}~\cite{azar}] There exists a $9.6$-competitive non-adaptive mechanism for laminar matroids.
\end{theorem}

\begin{theorem}[\textbf{Truncated Partition Matroid}~\cite{caramanis}] There exists an $8$-competitive non-adaptive mechanism for truncated partition matroids.
\end{theorem}

\section{Comparison to known results}

Our results for cographic matroids and  $k$-column sparse matroids are obtained through modifications of the arguments in~\cite{soto2013matroid} and~\cite{chawla2020non}, respectively. The results on regular matroids and minor-closed families of matroids follow the approach outlined in~\cite{tonypeter2020matroid} for the secretary problem. As necessary building blocks we use our results for cographic and $2$-column sparse matroids. Note that a biggest challenge for us is the compatibility of non-adaptive thresholds with contractions. Indeed, standard tools for deriving mechanisms for contraction minors need subconstraints, while subconstraints are not permitted in non-adaptive mechanisms. To obtain our results, we resolve this issue only in the context of matroids representable over finite fields, see arguments in Lemma~\ref{lm:representable_matroids}. It would be interesting to see whether analogous results for contraction minors hold with no assumption about representability over finite fields.

\section{Preliminaries}

In this paper, we consider the matroid prophet inequality problem, where
items arrive online in adversarial order. Here, the adversary knows the distributions of all $X_1$, $X_2$, \ldots, $X_n$ and knows the gambler's mechanism, but the realization of $X_1$, $X_2$, \ldots, $X_n$ is not known to the adversary. Based on the available information, the adversary can decide on the order in which items and their values are observed by the gambler. 

\subsection{Prophet inequality}

\begin{definition} Let $M$ be a matroid  on the ground set $[n]:=\{1,\ldots,n\}$, where $[n]$ corresponds to $n$ items. Let $\vec{X}:= (X_1, \ldots, X_n)$ be non-negative independent random variables representing the values of these $n$ items.
  \begin{itemize}
  \item For every subset of items $S\subseteq [n]$ we define its weight as follows
  \[
  w(S):= \sum_{i \in S} X_i.
  \]
  \item Let $\proph_M$ be the random variable corresponding to the value obtained by the prophet
  \[\proph_M:= \max_{S\in \mathcal I(M)} w(S)\,,\]
  where $\mathcal I (M)$ is a collection of independent sets for $M$.
  \item Let  $\eproph_M$ be the expectation of the value obtained by prophet
    \[\eproph_M:=E[\proph_M]\,.\]
  \end{itemize}
\end{definition}

\begin{definition} Let us be given a number $\alpha>0$.
\begin{itemize}
    \item We call a mechanism \emph{$\alpha$-competitive} (alternatively, we say that the mechanism \emph{guarantees an $\alpha$-approximation}) on the matroid $M$ if the expected value obtained by the gambler via this mechanism is at least $\frac{1}{\alpha}\eproph_M$. 
    
    \item We call a mechanism \emph{$\alpha$-competitive} (alternatively, we say that the mechanism \emph{guarantees an $\alpha$-approximation}) on the matroid class $\mathcal M$ if this mechanism is $\alpha$-competitive for every matroid $M\in \mathcal M$.
\end{itemize}
\end{definition}

\subsection{Non-adaptive mechanism}
We say that a mechanism is \emph{non-adaptive} if it has the following structure:
\begin{itemize}
    \item Given the distributions of $\vec{X} = (X_1, \ldots, X_n)$, the mechanism determines the values of \emph{thresholds} $\vec{T} = (T_1, \ldots, T_n)$, where each $T_i$, $i\in [n]$ is a real number or $+\infty$.
    \item If the value of item $i\in [n]$ is observed, the gambler accepts the item $i$ if and only both conditions hold:
    \begin{enumerate}
        \item the observed value of $X_i$ is at least $T_i$
        \item \label{enum: nonadaptive independence property} the item $i$ together with all previously selected items forms an independent set with respect to the matroid $M$.
    \end{enumerate}
\end{itemize}

Note, that a non-adaptive mechanism does not change thresholds during its course. So, none of the thresholds depends on the realization of~$\vec{X} = (X_1, \ldots, X_n)$.

Another crucial feature of a non-adaptive mechanism is that the mechanism works only with the original matroid~$M$. A non-adaptive mechanism does not allow us to define a new matroid $M'$, such that a set of items is independent in $M'$ only if it is independent in $M$, and modify the condition~\eqref{enum: nonadaptive independence property} based on $M'$.

In this work, we focus on non-adaptive mechanisms. From here
and later we use the term \emph{mechanism} to
refer to non-adaptive mechanisms \textbf{exclusively}.

\begin{remark} In this work, non-adaptive mechanisms are allowed to make non-deterministic decisions. Hence,  we allow a non-adaptive mechanism to construct the thresholds $\vec{T} = (T_1, \ldots, T_n)$ non-deterministically. 

To measure the performance of such a mechanism we use the expected total value, where the expectation is taken not only with respect to $\vec{X} = (X_1, \ldots, X_n)$ but also  with respect to $\vec{T} = (T_1, \ldots, T_n)$.

\end{remark}

\subsection{Matroids}

We provide a review of matroids here. 
Experienced readers should consider skipping or skimming this section. For further results about matroids, consider consulting~\cite{oxley}.

 A \emph{matroid}~$M = (E, \mathcal{S})$ is a pair of a finite \emph{ground set}~$E$ and a collection~$\mathcal{S}\subseteq 2^E$ of \emph{independent} sets.
The collection~$\mathcal{S}\subseteq 2^E$ of subsets of~$E$ satisfies the following conditions:
\begin{itemize}
    \item[(i)] Empty set is an independent set, so $\varnothing \in \mathcal{S}$.
    \item[(ii)] The collection $\mathcal S$ is closed with respect to taking subsets, so for all $A \subseteq B \subseteq E$ if $B$ is in $\mathcal{S}$ then  $A$ is in $\mathcal{S}.$  \item[(iii)] The collection $\mathcal S$ satisfies so called \emph{augmenation property}. In other words, for all $A, B \subseteq E$ such that $A, B \in \mathcal{S}$ and $|A| > |B|$, 
    there exists $c \in A \setminus B$ such that $B \cup \{c\} \in \mathcal{S}$.
\end{itemize}

A subset of $E$ is called \emph{dependent} if it is not in $\mathcal{S}$. The inclusion-maximal independent sets are called \emph{bases} and the inclusion-minimal dependent sets are called \emph{circuits}. For every two bases, their cardinalities are equal: for every bases $A$ and $B$ of $M$ we have $|A|=|B|$. A rank function for the matroid $M$ is a function $r_M: 2^E \to \mathbb{N}$ such that for every $A\subseteq E$ the value $r_M(A)$ equals the cardinality of an inclusion-maximal independent subset of $A$. In the cases when the choice of the matroid is clear from the context, we write $r$ instead of $r_M$. 

Given a matroid $M$, we can define the \emph{dual matroid} $M^*$ over the same ground set $E$. A set $A$ is independent for matroid $M^*$ if and only if $E\setminus A$ contains a basis of $M$. An element $c\in E$ is called a \emph{loop} in $M$ if $r_M(c)=0$. An element $c\in E$ is called a \emph{free element} in $M$ if $r_{M^*}(c)=0$. To put it another way, an element $c$ is free, if and only if for every set $A$, which is independent in $M$, $A\cup \{c\}$ is also independent in $M$. We say that elements $c$ and $d\in E$ are \emph{parallel} in matroid $M$, denoted by $c\parallel d$, if $r_M(c)=r_M(d)=r_M(\{c,d\})=1$. One can show that ``being parallel'' defines an equivalence relation on the non-loop elements of $M$. A matroid is called \emph{simple} if it has no loops and no parallel edges.

Let $M = (E, \mathcal{S})$ be a matroid and $A \subseteq E$. The \emph{contraction} of $M$ by $A$, denoted as $M/A$, is a matroid over ground set $E \setminus A$ with the following independent sets
\[\{S\subseteq E\setminus A\,:\, S\cup A' \in \mathcal {S}\}\,,
\]
where $A'$ is an inclusion-maximal independent subset of $A$.

The restriction of $M$ to $A$, denoted as $M \mid_{A}$ or $M \setminus \overline{A}$, is a matroid over the ground set $A$ where a set $S\subseteq A$ is independent in $M \mid_{A}$ if and only if it is independent in $M$.

A matroid $M'$ is called a \emph{simple version} of $M$ if $M'$ is obtained from $M$ by deleting all loops and contracting every parallel class of elements into a single element.

For matroids $M$, $N$, we say that $N$ is a \emph{minor} of $M = (E, \mathcal{S})$
  if $N$ is isomorphic to $M{/}A{\setminus}B$ for some disjoint sets $A, B\subseteq E$. A matroid class $\mathcal{M}$ is called \emph{minor-closed} if for any $M \in \mathcal{M}$ every minor of $M$  is also in $\mathcal M$.

Let us now list some of the classical examples of matroids, which were extensively studied in the context of various mathematical fields.

\begin{itemize}
\item A \emph{uniform} matroid $M = (E, \mathcal{S})$ of rank $k$ is matroid where \[\mathcal S:=\{A\subseteq E\,:\, |A|\leq k\}\,.\]
When $|E|=n$, we denote the uniform matroid of rank $k$ as $U_{k,n}$.
    \item A \emph{graphic} matroid over graph $G = (V, E)$ is a matroid $M = (E, \mathcal{S})$, where
    \[\mathcal S:=\{A\subseteq E\,:\, A \text{ is acyclic}\}\,.\]
      The graphic matroid over graph $G$ is denoted as $M(G)$. 
    \item A \emph{cographic} matroid  over graph $G = (V, E)$ is a dual matroid  $M = (E, \mathcal{S})$ to the graphic matroid over the same graph $G$. In this case we have
    \[\mathcal S:=\{A\subseteq E\,:\, (V,E\setminus A) \text{ has the same number of components as } (V,E)\}\,.\]
    \item A \emph{vector} matroid $M = (E, \mathcal{S})$ is a matroid such that there is
     a vector space~$V$ and a map $\phi: E\rightarrow V$  satisfying
    \[\mathcal S:=\{A\subseteq E\,:\, \text{multiset }\, \phi(A) \text{ is linearly independent}\}\,.\]
    
    Given a field $\mathbb F$, we say that $M$ is \emph{representable} over field~$\mathbb{F}$ if $M$ is isomorphic
  to the vector matroid where $V$ is a vector space over field~$\mathbb{F}$. 
  
  A matroid is called \emph{regular} if it is representable over every field. A matroid is called \emph{binary} if it is representable over $\mathbb{F}_2$.

  \item A \emph{$k$-column sparse} matroid $M = (E, \mathcal{S})$ is a matroid such that there is
     a  field $\mathbb F$ and dimension $m$ and a map $\phi: E\rightarrow \mathbb{F}^m$  such that 
         \[\mathcal S:=\{A\subseteq E\,:\,\text{multiset }\, \phi(A) \text{ is linearly independent over }\mathbb F\}\,;\]
         and moreover $\phi(c)\in \mathbb{F}^m$ has at most $k$ nonzero coordinates for every $c\in E$.
         
         \item A \emph{$\gamma$-sparse} matroid $M = (E, \mathcal{S})$ is a matroid such that the inequality $|S| \le \gamma r_M(S)$ holds for every $S \subseteq E$.
         
         \item A \emph{laminar} matroid $M = (E, \mathcal{S})$ is a matroid such that there exists a laminar family $\mathcal F$ over the ground set $E$ and there are numbers $c_F\in \mathbb{N}$, $F\in \mathcal F$ such that
         \[\mathcal S:=\{A\subseteq E\,:\,|A\cap F|\leq c_F\text{ for every }F\in \mathcal F\}\,.\]
         Moreover, if $\mathcal F=\{E, E_1,\ldots, E_k\}$, where $E_1$, \ldots, $E_k$ form a partition of the ground set $E$, then $M$ is called a \emph{truncated partition matroid}.
         Recall, that a family $\mathcal F$ is called \emph{laminar} if for every $A$, $B\in \mathcal F$ we have $A\subseteq B$ or $B\subseteq A$ or $A\cap B=\varnothing$.
\end{itemize}

Given a matroid $M = (E, \mathcal{S})$ we can define the corresponding polytope $P_M \subseteq \mathbb{R}^E$ as the convex hull of points corresponding to the characteristic vectors of independent sets.  The polytope $P_M$ is known to admit the following outer description~\cite{schrijver2003combinatorial}.
\begin {align*}  P_M = \{ x \in \mathbb{R}^{E} \,:\,\, &x \geq 0\text{ and }\\
  &x(S) \leq r_M(S)\text{  for every  }S\subseteq E\}\,,
\end{align*}
where $x(S)$ stands for $\sum_{c\in S} x_c$.

For a matroid $M = (E, \mathcal{S})$ and  a set $A\subseteq E$ we can define the \emph{closure} of $A$ as the following set 
\[\cl_M(A):=\{ c \in E \mid r_M(A \cup \{c\}) = r_M(A)\}\,.\]

For a matroid $M = (E, \mathcal{S})$, we call the following function $\sqcap_M \colon E \times E \to \mathbb{Z}$ a \emph{local connectivity} function 
\[\sqcap_M(X, Y) = r(X) + r(Y) - r(X \cup Y)\,.\]
The following function $\lambda_M \colon E \to \mathbb{Z}_{\ge 0}$ is called a \emph{connectivity} function 
\[\lambda_M(X) := \sqcap_M(X, E \setminus X) = r(X) + r(E \setminus X) - r(E)\,.\]

Informally, connectivity functions measure dependence with respect to the matroid between parts of the ground set.
To illustrate it, let us consider the connectivity function for vector matroids. Suppose $M = (E, \mathcal{S})$ is a vector matroid defined by a vector space $V$ and a map $\phi: E\rightarrow V$. 
Then we have  \begin{align*}
\lambda_M(S) = &r(S) + r(E \setminus S) - r(E) =\\
&\dim(\Span \phi(S)) + \dim(\Span \phi(E \setminus S)) - \dim (\phi(E)) =\\
&\dim\left(\left(\Span \phi(S)\right) \cap (\Span \phi(E \setminus S))\right)\,.
\end{align*}

\subsection{Ex-ante relaxation to the matroid polytope}

The goal of ex-ante relaxation \cite{feldman2016online} or \cite{chawla2020non} is to reduce the original problem to the problem where item values are distributed as independent Bernoulli random variables. Note, that both problems are using the same matroid.

In the original problem item values $\vec{X} = (X_1, \ldots, X_n)$ are independent random variables with known distributions. For $i \in [n]$ let $F_i$ be the cumulative distribution function of $X_i$. The reduction of the original problem to a new problem is done using a point $p$ in the matroid polytope $P_M$. Let us first show that there is a point $p\in P_M$ with properties that prove to be desirable later following the argumentation in \cite{chawla2020non}.

\begin{lemma}\label{lm:ex_ante} Given a matroid $M$ over the ground set $[n]$ and
  random variables $\vec{X} = (X_1, \ldots, X_n)$, there exists $p \in P_M$ such that \[\eproph_M \le \sum_{i=1}^n p_i t_i\,,\]
  where $t_i:= E[X_i \mid X_i \ge F_i^{-1}(1-p_i)]$ for every $i\in [n]$\footnote{
Here, we assume that for every $i\in [n]$ the event $X_i = F_i^{-1}(1-p_i)$ happens with the zero probability, which is true for all continuous distributions. In case of discrete distributions one needs to introduce appropriate tie-breaking.}.
\end{lemma}
\begin{proof} Let $I_{opt}$ be a random variable indicating an optimal independent set in~$M$ with respect to $\vec{X} = (X_1, \ldots, X_n)$. In case when for some realization of~$\vec{X} = (X_1, \ldots, X_n)$ there are several optimal independent sets, $I_{opt}$ can be selected as any of these sets. For $i\in [n]$, let $p_i$ be the probability that element $i$ is in $I_{opt}$.
Note that $p=(p_1,\ldots, p_m)$  is a convex combination of independent sets of $M$, and so lies in $P_M$.

Due to $\eproph_M=E[\sum_{i \in I_{opt}} X_i]$, it remains to show that
\[E[\sum_{i \in I_{opt}} X_i] \le \sum_{i=1}^n p_i t_i\,.\]
\relax
We have
\[E[\sum_{i \in I_{opt}} X_i] = \sum_{i=1}^n P[i \in I_{opt}] E[X_i \mid i \in I_{opt}] = \sum_{i=1}^n p_i E[X_i \mid i \in I_{opt}]\,.\]
\relax
For every $i\in[n]$ we have that $t_i$ and $E[X_i \mid i \in I_{opt}]$ are expectations of the same random variable $X_i$ but conditioned on the event $X_i \ge F_i^{-1}(1-p_i)$ and on the event $i \in I_{opt}$, respectively. Note, that the probability of both these events equals $p_i$. However, the expectation of $X_i$ conditioned on $X_i \ge F_i^{-1}(1-p_i)$ is the ``largest'' conditional expectation of $X_i$ on an event of probability $p_i$.
Thus, we have $p_i E[X_i \mid i \in I_{opt}] \le p_i t_i$ for every $i\in [n]$ and so we get the desired inequality
\[\sum_{i=1}^n p_i E[X_i \mid i \in I_{opt}] \le \sum_{i=1}^n p_i t_i\,.\]
\end{proof}

Let us show how one can use the point $p=(p_1,\ldots, p_n)$ guaranteed by Lemma~\ref{lm:ex_ante} to reduce the original problem. Let us define independent Bernoulli  random variables $\vec{X}' = (X'_1, \ldots, X'_n)$ as follows, for each $i\in [n]$
\[X'_i = \begin{cases}
t_i & \text{with probability~~} p_i \\
0 & \text{with probability~~}  1 - p_i\,,
\end{cases}\]
where $t_i: = E[X_i \mid X_i \ge F_i^{-1}(1-p_i)]$.

Let us assume that we have a non-adaptive mechanism for the original matroid $M$ and item values $\vec{X}' = (X'_1, \ldots, X'_n)$, which sets nonnegative thresholds $\vec{T}'=(T'_1,\ldots, T'_n)$. By definition of $\vec{X}' = (X'_1, \ldots, X'_n)$, for every $i\in [n]$ the exact value of $T'_i$ is not relevant per se, but it is crucial whether $t_i\geq T'_i$ or $t_i<T'_i$. If for some $i\in [n]$ we have $T'_i> t_i$ then this item $i$ is ``inactive'' and so is never selected by the gambler working with $M$ and  $\vec{X}' = (X'_1, \ldots, X'_n)$.

The key is to construct a non-adaptive mechanism for the original matroid $M$ and item values $\vec{X}' = (X'_1, \ldots, X'_n)$ with positive thresholds $\vec{T}'=(T'_1,\ldots, T'_n)$ such that for each item $i\in [n]$ the probability that $i$ is selected by the gambler is at least $\alpha p_i$. Now we can use such a non-adaptive mechanism for the original matroid $M$ and item values $\vec{X}' = (X'_1, \ldots, X'_n)$   to construct a non-adaptive $\alpha$-competitive mechanism  for the same matroid $M$ and random variables~$\vec{X} = (X_1, \ldots, X_n)$. Let us define the  thresholds $\vec{T}=(T_1,\ldots, T_n)$ as follows, for every $i\in [n]$
\[
T_i:=\begin{cases}
+\infty& \text{if }t_i<T'_i\\
F_i^{-1}(1-p_i)& \text{otherwise}\,.
\end{cases}
\]

To see that the thresholds $\vec{T}=(T_1,\ldots, T_n)$ lead to an  $\alpha$-competitive mechanism for $M$ and~$\vec{X} = (X_1, \ldots, X_n)$, let us couple random variables $X'_i$ with random variables $X_i$ as follows
\[
X'_i:=\begin{cases}
t_i& \text{if }X_i\geq F_i^{-1}(1-p_i)\\
0& \text{otherwise}\,.
\end{cases}
\]
Note that $\vec{X}' = (X'_1, \ldots, X'_n)$ are independent Bernoulli random variables, where for each $i\in [n]$ the variable $X'_i$ equals $t_i$ with probability $p_i$ and equals $0$ with probability $1-p_i$. When $\vec{X}'$ are coupled with $\vec{X}$ this way, $X_i$ and $X'_i$ have the same expected value when conditioned
on $X'_i$ being $t_i$. The mechanism with thresholds $\vec{T}$ selects an item $i\in [n]$ when run for $\vec{X}$ only if the mechanism with thresholds $\vec{T}'$ selects the item $i$ when run for $\vec{X}'$. Moreover, for both of these algorithms, conditionally on the event that the item $i$ is selected the expected value of $i$ equals $t_i$. Now, $\alpha$-competitiveness guarantee of the thresholds $\vec{T}$ for $M$ and $\vec{X}$ follows from Lemma~\ref{lm:ex_ante}.

\section{Graphic and $k$-column sparse matroids}

First, we construct a $16$-competitive non-adaptive mechanism for graphic matroids without parallel edges. Our construction is done through the ex-ante relaxation to the matroid polytope, following the works in \cite{feldman2016online} or \cite{chawla2020non}. Later, we present a constant-competitive non-adaptive mechanism for $k$-column sparse matroids whenever $k$ is constant.

\subsection{Graphic matroids}
\label{sec:graphic_matroid}

Now  we are ready to provide a $16$-competitive non-adaptive mechanism for graphic matroid. The provided mechanism is essentially the one constructed in~\cite{chawla2020non} but with saving a factor of $2$ in the guarantee, which is achieved by rescaling the point from the matroid polytope by $2$ and not by $4$.

Let us be given a simple graph $G = (V, E)$ and let us consider the corresponding graphic matroid $M$ over the ground set $E$. Recall that a subset of~$E$ is independent with respect to $M$ if and only if it is acyclic in $G$. Let us also assume that the graph $G$ has $n$ edges and so $E = \{e_1, e_2, \ldots, e_n\}$.

\begin{lemma}\label{lm:orient} Let $p=(p_1,\ldots, p_n)$ be a point in the polytope $P_M$. Thus we assume that for every $i\in [n]$ the coordinate $p_i$ of $p$ corresponds to the edge $e_i$.
Then there exists an orientation of edges $E = \{e_1, e_2, \ldots, e_n\}$ in the graph $G = (V, E)$ such that for every vertex $v\in V$ we have $\sum_{i\in [n]:e_i \in \delta^-(v)}p_i\leq 2$.
\end{lemma}

\begin{proof}
Observe that the average degree of a vertex in a forest on $|V|$ vertices is at most $(2|V| - 2)/|V| = 2 - 1/|V| \le 2$.

Let us use this fact to prove the desired statement by induction on the number of vertices in the graph $G$.

If the graph $G$ has at most two vertices then the orientation is trivial. Otherwise, since $p$ is a convex combination of points corresponding to forests in $G$, we have that the average of the value $\sum_{i\in [n]:e_i \in \delta(v)}p_i$ over all vertices $v\in V$ is at most $2$. Thus there exists a vertex $v\in V$ such that we have
$\sum_{i\in [n]:e_i \in \delta(v)}p_i\leq 2$. We orient all edges incident to $v$ as edges in $\delta^{-}(v)$, so these edges are incoming with respect to $v$. Then we remove the vertex $v$ and all edges incident to it and orient the remaining edges according to the orientation guaranteed by the inductive hypothesis.
\end{proof}

Now we present an algorithm for graphic matroids of simple graphs. 
\begin{algorithm}
\begin{algorithmic}[1]
\caption{A non-adaptive 16-competitive mechanisms for graphic matroids of a simple graph}\label{alg:graphic}
\State Let $p$ be a point in the polytope $P_M$ so that the statement of Lemma~\ref{lm:ex_ante} is satisfied.

\State Let the edges of the original graph $G=(V,E)$  be oriented so that the statement of Lemma~\ref{lm:orient} is satisfied.

\State For every edge $e_i \in E$, $i\in [n]$, mark the edge $e_i$ as ``discarded" independently at random with probability $1/2$. 

\State Select a cut $S\subseteq V$ uniformly at random, mark all edges not in $[S;\overline{S}]$ as ``discarded". Here, $[S;\overline{S}]$ stands for the set of edges which are oriented such that their tail is in $S$ and their head is in $\overline S$.

\State Set thresholds $\vec{T}=(T_1,\ldots, T_n)$ as follows, for each $i\in [n]$
\[
T_i:=\begin{cases}
+\infty& \text{if }e_i\text{  is ``discarded''}\\
F_i^{-1}(1-p_i)& \text{otherwise}\,.
\end{cases}
\]

\end{algorithmic}
\end{algorithm}

\begin{lemma}\label{lm:good_chance} For every $i\in[n]$, we have 
\[
P[e_i \text{ is selected}\mid X_i\geq T_i \text{  and  } e_i  \text{ is not ``discarded''}]\geq 1/2\,.
\]
\end{lemma}

\begin{proof}
Let us assume that the vertex $v$ is the head of the  oriented edge $e_i$. Let us also assume that $e_i$ is not marked as ``discarded'' and $X_i\geq T_i$.

Since the edge $e_i$ is not ``discarded'',  the edge $e_i$ is in the selected set $[S; \overline{S}]$. Hence, every not ``discarded'' edge incident to $v$ has the vertex $v$ as its head.

Thus, as long as no other edge with the head at the vertex $v$ is selected by the gambler, the gambler has to select $e_i$. We claim, that with probability at least $1/2$ no other edge with the head at $v$ was selected by the gambler. 

Let $I$ be the event indicating that "the gambler selected an edge $e_j$, $j\neq i$ such that $v$ is the head of $e_j$", in other words  ``there is $j\in [n]$, $j\neq i$ such that $v$ is the head of $e_j$ and $X_j\geq T_j$ and $e_j$ is not ``discarded''". Let $J$ indicate the event that "$e_i$ is not marked as ``discarded'' after the selection of the cut", in other words, "the head of $e_i$ is in $\overline S$ and  the tail of $e_i$ is in $S$". 

Let us show 
\[P[I\mid J]\leq 1/2\,.\]
By the union bound, we have
\[P[I \mid J] \le \sum_{j\in [n]\setminus\{i\}: e_j \in \delta^-(v)} P[X_j\geq T_j\text{  and
$e_{j}$ is not ``discarded''} \mid J]\]
Note that for each edge $e_j \in \delta^-(v)$ we have $P[X_j\geq T_j|J]=p_j$ and we also have $P[\text{$e_{j}$ is not ``discarded''} |J]=1/4$. Note that any edge  is not ``discarded'' in Step~3 of Algorithm~\ref{alg:graphic} with probability $1/2$, and not ``discarded'' in Step~4 of Algorithm~\ref{alg:graphic} with probability $1/4$. However, since the probabilities are with respect to the edge $e_j \in \delta^-(v)$ and are  counted conditioned on the event $J$, the conditioned probability of not being ``discarded''  in Step~4 of Algorithm~\ref{alg:graphic} is $1/2$.
 Moreover, even conditioned on $J$ the events  "$X_j\geq T_j$" and "$e_{j}$ is not ``discarded''" are independent events. Thus we have
\begin{align*} \sum_{j\in [n]\setminus\{i\}: e_j \in \delta^-(v)} P[X_j\geq T_j\text{  and
$e_{j}$ is not ``discarded''} \mid J]\leq \\ \sum_{j\in [n]\setminus\{i\}: e_j \in \delta^-(v)}p_j/4\le 1/2\,,
\end{align*}
where the last inequality follows from the orientation.
\end{proof}
  
We are ready to prove Theorem~\ref{thm:graphic16} by showing that  Algorithm~\ref{alg:graphic} is a $16$-competitive for graphic matroids without parallel edges.

\begin{proof}[Proof of Theorem~\ref{thm:graphic16}]

By Lemma~\ref{lm:good_chance} for every $i\in [n]$ the probability of edge $e_i$ being accepted conditional on $X_i\geq T_i$ and being not ``discarded'' is at least $1/2$.

Overall, the probability of edge $e_i$ being accepted is at
least $\frac{1}{16} p_i$. Thus mechanism guarantees at least
$\sum_{i=1}^n \frac{1}{16} p_i t_i$ of the expected total value. By Lemma~\ref{lm:ex_ante}, we have $\sum_{i\in[n]} \frac{1}{16} p_i t_i \ge \frac{1}{16} \eproph_M$, finishing the proof.
\end{proof}

\subsection{$k$-column sparse matroids}

There are known constant-competitive mechanisms
for $k$-column sparse matroids in the context of the secretary problem~\cite{soto2013matroid}. However they do not immediately lead to a non-adaptive mechanism of constant competitiveness guarantee. The reason for that are not the updated thresholds but implicit changes to the considered matroid.

Here, we present a constant competitive mechanism for $k$-column sparse matroid class for each constant $k$. Note, graphic matroids form a subclass of  $2$-column sparse matroids. Because of their significance,  $2$-column sparse matroids
are also known in literature as \emph{represented frame matroids}. Later,
we use $2$-column sparse matroids to prove results in Section~\ref{sec:minor_closed}.

Suppose $M$ is a $k$-column sparse matroid over
  field $\mathbb{F}$. In this section, we prove that there exists a $(2^{k+2} k)$-competitive mechanism for $M$.

 Suppose a $k$-sparse representation of $M=(E,\mathcal S)$ is defined by a map $\phi:E\rightarrow\mathbb{F}^d$. Note, if for some element $t\in E$ the vector $\phi(t)$ is a zero vector then $c$ is a loop and therefore can be removed from consideration.

  Now we consider an undirected hyper-multigraph $G$ with vertex
  set $[d]$. Each matroid element $t \in E$ induces
  a hyperedge $e_t$ in this graph between non-zero coordinates of $\phi(t)$.
  Formally, the hyperedge $e_t$ is defined as follows $e_t := \{ i \in [d] \,:\, \phi(t)_i \ne 0\}$.
  We say that  a vertex $i\in [d]$ of the hyper-multigraph $G$ is incident to every edge $e$ of $G$ such that $i\in e$. For a vertex $i\in[d]$ we denote the collection of  incident hyperedges by $\delta(i)$. The degree of a vertex $i$ in the hyper-multigraph $G$ equals $|\delta(i)|$.     
  
    \begin{claim}\label{claim:avg2} Suppose $I$ is an independent set of the matroid $M$. Then the average degree of a vertex is at most $k$ when one considers the hyper-multigraph with vertices $[d]$ and hyperedges $\{e_t \,:\, t \in I\}$.
  \end{claim}

  \begin{proof} Observe that $|I| \le d$ because having more than
    $d$ vectors in $d$-dimensional vector space $\mathbb{F}^d$ leads to a a linear dependency.

    Since $M$ is $k$-column sparse, we have that every edge in $\{e_t \,:\, t \in I\}$ is incident to at most $k$ vertices in $[d]$. Hence, the total degree is at most $k d$ and thus the average degree of a vertex is at most $k$.
  \end{proof}
  
  Now we consider \emph{orientations} of the graph $G$.
  An orientation of the graph $G$ is a function $\vphi$ which maps every edge $e_t$ into one vertex of $G$ incident to $e_t$.
  We call $\vphi(e_t)$ to be the head of the edge $e_t$, and all other vertices, if any,
  to be tails. For every vertex $i\in [d]$ we denote the set of incoming edges by
  $\delta^{-}(i)$, formally $\delta^{-}(i) = \{ e_t \,:\, \vphi(e_t) = i,\, t\in E \}$.

  \begin{lemma}\label{lm:orient2} Let $p$ be a point in the polytope $P_M$. We assume that for every $t\in E$, the coordinate $p_t$ of $p$ corresponds to the element $t$.
Then there exists an orientation $\vphi$ of hyperedges in the hyper-mulrigraph $G$  such that for every vertex $i\in [d]$ we have $\sum_{t\in E : e_t \in \delta^{-}(i)} p_t \le k$.
  \end{lemma}
  
  The proof of Lemma~\ref{lm:orient2} is analogous to the proof of Lemma~\ref{lm:orient}.
 Now let us describe an algorithm for $k$-column sparse matroids.
  \begin{algorithm}[h!]
\begin{algorithmic}[1]
\caption{A non-adaptive $2^{k+2}k$-competitive mechanisms for $k$-column sparse matroids}\label{alg:ksparse}
\State Let $p$ be a point in the polytope $P_M$ so that the statement of Lemma~\ref{lm:ex_ante} is satisfied.

\State Let the edges of the hyper-multigraph $G$  be oriented so that the statement of Lemma~\ref{lm:orient2} is satisfied.

\State For every edge $e_i \in E$, $i\in [n]$, mark the edge $e_i$ as ``discarded" independently at random with probability $1 - \frac{1}{2k}$.

\State Select a cut $S\subseteq [d]$ uniformly at random, mark all edges not in $[S;\overline{S}]$ as ``discarded''. Here, $[S;\overline{S}]$ stands for the set of edges which are oriented such that  all their tails are in $S$ and their head is in $\overline S$.    In particular, for $t\in E$ we say that $e_t$ lies in a cut $[S; \overline{S}]$ with respect to the orientation
  $\vphi$ if $\vphi(e_t) \in \overline{S}$ and for every
  $i \in e_t \setminus \{\vphi(e_t)\}$ we have $i\in S$.

\State Set thresholds $\{T_t\,:\, t\in E\}$ as follows, for each $t\in E$
\[
T_t:=\begin{cases}
+\infty& \text{if }t\text{  is ``discarded''}\\
F_t^{-1}(1-p_t)& \text{otherwise}\,.
\end{cases}
\]

\end{algorithmic}
\end{algorithm}

    \begin{lemma}\label{lm:good_chance2} For every $t\in E$ we have 
\[
P[t \text{ is selected}\mid X_t\geq T_t \text{  and  } t  \text{ is not ``discarded''}]\geq 1/2\,.
\]
\end{lemma}

  \begin{proof}     
    Note that item $t\in E$ is accepted whenever $X_t\geq T_t$ and no other item was selected from
   non-discarded edges in $\delta^{-}(\vphi(t))$. By the union bound, for every event $J$ we can upper bound the probability that 
    \begin{align*}
    P[\text{there $j\in E\setminus\{t\}$ such that $j$ is selected and  $e_j\in \delta^{-}(\vphi(t))$}\mid J] \le
    \\ \sum_{j\in E\setminus\{t\}: e_j\in \delta^{-}(\vphi(t))} P[\text{$e_j$ is not ``discarded'' and } X_j\geq T_j\mid J]\,.
    \end{align*}
Let $J$ indicate the event that "$e_t$ is not marked as ``discarded'' after the selection of the cut".  Then for each $j\in E\setminus\{t\}$ we have $P[\text{$e_j$ is not ``discarded'' and } X_j\geq T_j\mid J]\leq \frac{1}{2k}p_j$.
By Lemma~\ref{lm:orient2},  we have $\sum_{j\in E: e_j\in \delta^{-}(\vphi(t))} p_j \le k$, leading to the desired inequality.
  \end{proof}

Note that the proof of Lemma~\ref{lm:good_chance2} is analogous to the proof of Lemma~\ref{lm:good_chance}.
We are ready to prove Theorem by showing that the Algorithm~\ref{alg:ksparse} is a $2^{k+2} k$-competitive for $k$-column sparse matroids.

\begin{proof}[Proof of Theorem~\ref{thm:k_column_sparse}]

For every item $t\in E$ we have $P[X_t\geq T_t]=p_t$ and $P[t\text{  is not ``discarded''}]\geq\frac{1}{2^{k+1} k}$.
  By Lemma~\ref{lm:good_chance2}, we have that with probability at least $1/2$ the  item $t$ is selected when it is not ``discarded'' and $X_t\geq T_t$. Thus the expected total value of Algorithm~\ref{alg:ksparse} is at least $\sum_{j\in E} \frac{1}{2^{k+2} k} p_j t_j$ which is at least $\frac{1}{2^{k+2} k} \eproph_M$ by Lemma \ref{lm:ex_ante}.
  
  \end{proof}

\section{Cographic and gamma-sparse matroids}

\subsection{Cographic matroids}

Let us revisit a mechanism of Soto~\cite{soto2013matroid} for the cographic matroid secretary problem which is based on the following corollary of Edmond's matroid partitioning theorem~\cite{edmonds1965minimum}. This mechanism leads to a non-adaptive mechanism for cographic matroids.

\begin{proposition}\label{prop:threeconnected} Let $G=(V,E)$ be a three edge-connected graph.
  Then there exist spanning trees $H_1, H_2, H_3$ in $G$ such
  that the union of their complements contains all the edges $E$, i.e. $E=(E\setminus H_1)\cup (E\setminus H_2)\cup  (E\setminus H_3)$.
\end{proposition}

 \begin{algorithm}[h]
\begin{algorithmic}[1]
\caption{A non-adaptive $3$-competitive mechanisms for cographic matroids in the case of three edge-connectivity}\label{alg:cogrpahic3connected}
\State Let $H_1$, $H_2$ and $H_3$ be the spanning trees as in Proposition~\ref{lm:threecompet_cogr}.

\State Uniformly at random select a spanning tree $H^*$ from  $H_1$, $H_2$ and $H_3$. Set thresholds $\{T_e\,:\, e\in E\}$ as follows, for each $e\in E$
\[
T_e:=\begin{cases}
+\infty& \text{if }e\text{  is not in }H^*\\
0& \text{otherwise}\,.
\end{cases}
\]

\end{algorithmic}
\end{algorithm}

\begin{lemma}\label{lm:threecompet_cogr} Let $G=(V,E)$ be a three edge-connected graph and  let $M$ be the cographic matroid over $G$. Then Algorithm~\ref{alg:cogrpahic3connected}  is a $3$-competitive non-adaptive mechanism for the matroid~$M$.
\end{lemma}

\begin{proof}
The expected total value of the mechanism provided by Algorithm~\ref{alg:cogrpahic3connected} equals  $E[\sum_{e\in E\setminus H^*}X_e]$ which can be estimated as follows
\[E[\sum_{e\in E\setminus H^*}X_e]=\frac{1}{3}E[\sum_{i\in [3]}\sum_{e\in E\setminus H_i}X_e]\geq \frac{1}{3}E[\sum_{e\in E}X_e]\geq \frac{1}{3}\eproph_M\,.
\]
\end{proof}

 \begin{algorithm}[h]
\begin{algorithmic}[1]
\caption{A non-adaptive $6$-competitive mechanisms for cographic matroids}\label{alg:cogrpahic}

\State Delete all loops of $M$ to obtain a matroid $M'$. Remove all bridges from $G=(V,E)$ and obtain a graph $G'=(V', E')$.

\State Let $C_1$,\ldots, $C_k$ be equivalence classes of $M'$ with respect to the relation of being parallel. Construct the matroid $M''$ from $M'$ by contracting all but one edge in each  class $C_1$, $C_2$, \ldots, $C_k$. Note, that the ground set of $M''$ has $k$ elements and matroid $M''$ is the cographic matroid over a graph $G''$, where each connected component of $G''$ is three edge-connected. Abusing the notation we refer to the elements of the ground set of $M''$ as $C_1$, $C_2$, \ldots, $C_k$. 

\State Let $H_1$, $H_2$ and $H_3$ be forests in $G''$ such that the restriction of $H_1$, $H_2$ and $H_3$ to each connected component of $G''$ satisfies Proposition~\ref{lm:threecompet_cogr} for the respective connected component.

\State Uniformly at random select a forest $H^*$ from  $H_1$, $H_2$ and $H_3$.

\State For each $i \in [k]$ select thresholds $\overline{T}_e$, $e\in C_i$ according to Theorem~\ref{thm:twocompetitive} when the gambler is allowed to accept only one item of $C_i$ and the distributions of $X_e$, $e\in C_i$ are the same as original distributions of values for $e\in C_i$.

\State
Set thresholds $\{T_e\,:\, e\in E\}$ as follows, for each $e\in E$
\[
T_e:=\begin{cases}
\overline{T}_e  &\text{if }e\in C_i\text{ and }C_i\in H^* \text{ for some }i\in [k]\\
+\infty& \text{otherwise}\,.
\end{cases}
\]

\end{algorithmic}
\end{algorithm}

The next theorem provides a proof for Theorem~\ref{thm:cographic}.
 
\begin{theorem}Let $G=(V,E)$ be a graph and  let $M$ be the cographic matroid over $G$. Then Algorithm~\ref{alg:cogrpahic}  is a $6$-competitive non-adaptive mechanism for the matroid~$M$.

\end{theorem}

\begin{proof}
 We can assume that $G$ does not have bridges, because every such bridge
  is a loop in $M$. Thus these edges can be selected neither by the gambler nor by the prophet. So we can assume $G=G'$ and $M=M'$.

  In the case when each connected component of $G$ is three edge-connected, then Algorithm~\ref{alg:cogrpahic} runs Algorithm~\ref{alg:cogrpahic3connected} for each component to obtain a $3$-competitive non-adaptive mechanism.

  Otherwise, there is one or more pairs of edges $e$,$e'$ such that $\{e, e'\}$ corresponds to a cut in $G$. In this case, the edges $e$,$e'$ correspond to parallel elements of the cographic matroid $M$. 

Algorithm~\ref{alg:cogrpahic}  considers the partition of $E$ into classes of parallel elements $C_1$, $C_2$, \ldots, $C_k$. Let us construct the matroid $M''$ from $M$ by contracting all but one edge in each  class $C_1$, $C_2$, \ldots, $C_k$. Note, that the ground set of $M''$ has $k$ elements. Abusing the notation we refer to these elements of the ground set as $C_1$, $C_2$, \ldots, $C_k$. 
  The matroid $M''$ is isomorphic to the cographic matroid over a graph $G''$, where each connected component of $G''$ is three edge-connected. Following Lemma~\ref{lm:threecompet_cogr}, Algorithm~\ref{alg:cogrpahic} constructs forests $H_1$, $H_2$, $H_3$ for the graph~$G''$.
  
 So Algorithm~\ref{alg:cogrpahic} leads us to a $6$-competitive mechanism. Indeed, the prophet with $M$ and with the original distributions of $X_e$, $e\in E$
  performs exactly as the prophet with $M''$ and with the corresponding distributions of $X''_i:=\max_{e\in C_i}X_e$, $i\in [k]$. By selecting forests in  Algorithm~\ref{alg:cogrpahic} the gambler acheives in expectation $E[\sum_{i\in[k]}X''_i]/3$ when all classes $C_1$, $C_2$, \ldots, $C_k$ are singletons. However, for classes that are not singletons we need to take into account another $2$ approximation factor with respect to the prophet, who can  achieve the expected value  $E[X''_i]$ for each $i\in [k]$, while the gambler is guaranteed in expectation to achieve only $E[X''_i]/2$ for each $i\in [k]$.
\end{proof}

\subsection{Gamma-sparse matroids}

Let us also revisit a mechanism of Soto~\cite{soto2013matroid} for $\gamma$-sparse matroids to verify that it directly leads to  a non-adaptive mechanism.

\begin{theorem}\label{thm:gamma_sparse} Let $M = (E, \mathcal{S})$ be a $\gamma$-sparse matroid. There exists
  a $\gamma$-competitive non-adaptive mechanism for $M$.
\end{theorem}

\begin{proof} First observe that the point $x := \mathbb{1}{/}\gamma$ lies in the matroid polytope $P_M$. Indeed, it is non-negative and for every set $S \subseteq E(M)$ we have $x(S) = |S|{/}\gamma \le r_M(S)$.

  Then $x$ can be expressed as a convex combination of indicator variables corresponding to
  the independent sets of $M$. In other words, we have $x = \sum_{S \in \mathcal{S}} \alpha_S \mathbb{1}_{S}$ for
  some $\alpha \ge 0$, $\sum_{S \in \mathcal{S}} \alpha_S = 1$, where $\mathbb{1}_{S}$ refers to the characteristic vector of $S$.

  Now sample an independent set $S$ in matroid $M$ randomly with probability~$\alpha_S$. Let the gambler select all items in $S$ and let the gambler leave all the items not in $S$ unselected.

  If $X_e$ is the random variable corresponding to the weight of element $e \in E(M)$,
  then this mechanism results in a total expected value as follows
  \[
  \sum_{S \in \mathcal{S}} \alpha_S \sum_{e \in S} E[X_e] = \sum_{e \in E} (1/\gamma) E[X_e]
  = E[ \sum_{e \in E} X_e]/\gamma \ge \eproph/\gamma\,,
  \]
 finishing the proof.
\end{proof}

Observe that  Proposition \ref{prop:threeconnected} implies
  that for a three edge-connected  graph~$G$, the cographic matroid of $G$
  is $3$-sparse. Thus Lemma~\ref{lm:threecompet_cogr} is a corollary of
  Theorem~\ref{thm:gamma_sparse}.
  
  Similarly, for a planar graph $G$ the graphic matroid is $3$-sparse, leading us to the following corollary.
  
\begin{corollary} Let $G$ is a planar graph and let $M$ be the corresponding graphic matroid. There is  a $3$-competitive non-adaptive mechanism for $M$.
\end{corollary}

\section{Representable matroids}

Many results in the theory of matroids  make use of minors coming from restrictions and contractions. To get access to the toolbox provided by matroid theory, we need to understand how prophet inequality guarantees change when we consider minors.

\subsection{Preliminaries}

\begin{lemma}\label{lm:apx_of_restriction} Let $M$ be a matroid and let matroid $N$ be a restriction of the matroid~$M$. If there exists an $\alpha$-competitive non-adaptive mechanism on~$M$, then there is
an $\alpha$-competitive non-adaptive mechanism for $N$.
\end{lemma}

\begin{proof} To obtain a mechanism for the matroid $N$, we can impose thresholds $+\infty$ for the items that were removed from the ground set to obtain the restriction~$N$ from the matroid~$M$. The remaining items are assigned the same thresholds in both mechanisms.
\end{proof}

A similar result for contractions is harder to obtain in the case of non-adaptive mechanisms. Indeed, a straightforward approach would require us to impose the thresholds $+\infty$ for the contracted items, while using the given mechanism on the remaining items. Unfortunately, this would also require us to ``change" the underlying matroid, in other words a gambler might be forced to reject an item even though its value is over the assigned threshold and its addition to the currently selected items keeps the selected set independent with respect to $M$.

Because of this difficulty, in this work we provide a matching result for contractions only for matroids representable over a finite field. This result is sufficient for the purpose of this work.

\begin{lemma}\label{lm:representable_matroids} Let $M=(E, \mathcal S)$ be a matroid representable over the field $\mathbb{F}_p$ for some $p$. Let  $T \subseteq E$ be a subset of the ground set such that $\lambda_M(T) \le k$ for some~$k$.

  Then there exists $S \subseteq T$ so that every set that is independent in 
  $M\mid_S$ is also independent in $M{/}_{\overline{T}}$ and   \[\eproph_{M \mid_{S}} \ge \frac{1}{p^{k+1}} {\eproph_{M{/}_{\overline{T}}}}\,.\]
  Recall that $\overline{T}$ stands for the complement of $T$ with respect to the ground set $E$.
  
\end{lemma}

\begin{proof} Consider the representation of the matroid $M$ over $\mathbb{F}_p$. Let $\phi: E\rightarrow \mathbb{F}^m_p$  be the map describing the representation of $M$. 
  Thus, for every $S\subseteq E$ we have that the set $\phi(S)=\{\phi(e)\in \mathbb{F}_p^m\,:\, e\in S\}$ is independent over the field $\mathbb{F}_p$ if and only if $S$ is an independent set for the matroid $M$.
  
  Since $\lambda_M(T) \le k$ holds, by definition of $\lambda_M$ we have
  \[r_M(T) + r_M(\overline{T}) - r_M(E) \le k\,.\]
We have $r_M(R)=\dim\Span(\phi(R))$ for every $R\subseteq E$. Thus, we have
\[
\dim\Span\phi(E) = \dim\Span \phi(T) + \dim\Span \phi(\overline T) -
  \dim \left((\Span \phi(T))\cap (\Span \phi(\overline{T}))\right).
\]
and so
 \[
\dim \left((\Span \phi(T))\cap (\Span \phi(\overline{T})) \right) \le k\,.
\]
 Since we are working over the field $\mathbb{F}_p$, the linear space $L:=(\Span \phi(T))\cap (\Span \phi(\overline{T}))$ has at most $p^k$ vectors.
  Let $C$ be the orthogonal complement of the linear space $L$ in the space $\Span \phi(T)$. Thus, we can represent $\Span \phi(T)$ as  $L \oplus C$.
  For every vector~$v \in \Span \phi(T)$ we denote $v$ orthogonal projection to $L$ and $C$ by $v \mid_{L}$
 and $v \mid_{C}$, respectively.

 For each vector $a\in L$, define the set $T_a:= \{ t \in T \,:\, \phi(t) \mid_L = a,\, \phi(t) \ne a \}$. Note
  that by definition for every $a\in L$ we have $T_a\cap L=\varnothing$. Now let us
  select $a$ uniformly at random from $L$.

  \begin{claim} 
  
  $E_a [ \eproph_{M \mid_{T_a}}] \geq \frac{1}{p^k} \eproph_{M{/}_{\overline{T}}}$.
  \end{claim}

  \begin{proof} To prove the desired inequality, we prove the corresponding inequality for any realization of item values. From now on we consider the realization of item values fixed and thus we prove the following inequality
\[ 
E_a [ \proph_{M \mid_{T_a}}] \geq \frac{1}{p^k} \proph_{M{/}_{\overline{T}}}
\]

Let us consider the set $I_{opt}$ on which the prophet achieves $\proph_{M{/}_{\overline{T}}}$. Note   that the set $I_{opt}$ does not contain any item $e$ such that $\phi(e)$ is in $L$, because every such an item $e$ is a
    loop in $M{/}_{\overline{T}}$. Thus, the set $I_{opt}$ can be partitioned into sets $I_{opt, a}$, ${a \in L}$ where $I_{opt, a}$ is a subset of $T_a$.

    The set $I_{opt}$ is independent in $M{/}_{\overline{T}}$ and so $I_{opt}$ is also independent in $M$. Hence the sets $I_{opt, a}$, ${a \in L}$ are also independent in $M$. Thus for every $a \in L$, $\proph_{M \mid_{T_a}} \ge w(I_{opt, a})$. Then we have
\[
E_a [ \proph_{M \mid_{T_a}}] \ge \frac{\sum_{a \in L} w(I_{opt, a})}{|L|} = \frac{1}{|L|} w(I_{opt}) \ge \frac{1}{p^k} \proph_{M{/}_{\overline{T}}}\,,
\]
 finishing the proof of the claim. 
  
  \end{proof}

Let us now select $a^*\in L$ such that $ \eproph_{M \mid_{T_a}}$ is maximized.  By the previous claim, we have
  \[
  \proph_{M \mid_{T_{a^*}}} \ge \frac{1}{p^k} \proph_{M{/}_{\overline{T}}}\,.
  \]
  Now for every $c \in C$ define set
  $H_c := \{t \in T_{a^*} \,:\, (\phi(t) \mid_{C}) \cdot c = 1 \}$.
 Now let us
  select $c$ uniformly at random from $C$.

  \begin{claim} $E_c[\eproph_{M \mid_{H_c}} ] \geq \frac{1}{p} \eproph_{M{\mid}_{\overline{T}_{a^*}}}.$

  \end{claim}

  \begin{proof} To prove the desired inequality, we prove the corresponding inequality for any realization of item values. From now on we consider the realization of item values fixed and thus we prove the following inequality
  \[ 
E_c[\proph_{M \mid_{H_c}} ] \geq \frac{1}{p} \proph_{M{\mid}_{\overline{T}_{a^*}}}\,.
\]

    Let  $I_{opt}$ be the set corresponding to $\proph_{M \mid_{T_{a^*}}}$. Thus, we have that
    for every  $e \in I_{opt}$, $\phi(e)$ is not in $L$ and hence $\phi(e) \mid_{C}$ is not the zero vector. Due to $P_{c} [ c \cdot t = 1 ] = 1{/}p$, for every $t\in T_{a^*}$, we have 
    \[
    E_c[w(I_{opt}\cap H_c)]=\sum_{t\in I_{opt}}P_{c} [ c \cdot t = 1 ] w(t)= \frac{1}{p}\sum_{t\in I_{opt}}w(t)=\frac{1}{p}w(I_{opt})=\proph_{M \mid_{T_{a^*}}}\,.
    \]
    Finally, since $I_{opt}$ is independent in $M$ so is $I_{opt}\cap H_c$. Thus, we have 
    \[E_c[\proph_{M \mid_{H_c}} ] \geq \frac{1}{p} \proph_{M{\mid}_{\overline{T}_{a^*}}}\,,\]
    finishing the proof of the claim.
  \end{proof}

  Now let us select $c^*$ so that $\eproph_{M \mid_{H_c}}$ is maximized and let $S^* := H_{c^*}$.
  Then we have $\eproph(M \mid_{S^*}) \ge \frac{1}{p^{k+1}} \eproph(M{/}_{\overline{T}})$. 
  
  Finally, we need to show that every set independent in $M\mid_{S^*}$ is an  independent set in $M{/}_{\overline{T}}$.
  Suppose the contrary, i.e. there exists a set that is independent in $M\mid_{S^*}$ but is not  an independent set in $M{/}_{\overline{T}}$. Then $\Span S^*$ has a non-trivial intersection with
  $\Span \overline{T}$, suppose $x \in (\Span \phi(S^*)) \cup (\Span \phi(\overline{T}))$. Let us show that $x$ is a zero vector. Since $x\in \Span S^*$, we have $x=\sum_{s\in S^*} \alpha_s \phi(s)$ for some $\alpha_s\in \mathbb{F}_p$, $s\in S^*$.

  Let us consider the projections of $x$ on $C$ and $L$. Since $x\in \Span \phi(\overline{T})$ we have that $x$ lies in $L$ and so $x\mid_C$ is the zero vector. Thus $x\mid_C=\sum_{s\in S^*}\alpha_s (\phi(s)\mid_C)$ is the zero vector.
  
  Note that by definition, $\phi(s) \mid_L = a^*$   and  $c^*\cdot (\phi(s)\mid_C)=1$ hold for every~$s\in S^*$ . 
  Thus over the field $\mathbb{F}_p$ we have 
  \begin{align*}
  \sum_{s\in S^*}\alpha_s=\sum_{s\in S^*}\alpha_s (c^*\cdot (\phi(s)\mid_C))=
  c^*\cdot \left(\sum_{s\in S^*}\alpha_s (\phi(s)\mid_C)\right)=\\
   c^*\cdot \left(x\mid_C\right)=
  0\,.
  \end{align*}
 
 Now let us consider $x\mid_L$. We have
 \[
 x\mid_L=\sum_{s\in S^*}\alpha_s (\phi(s)\mid_L)=\left(\sum_{s\in S^*}\alpha_s\right) a^*\,,
 \]
 where the last expression equals the zero vector since $\sum_{s\in S^*}\alpha_s=0$.
 Thus we have a vector $x\in L\oplus C$ such that both projections $x\mid_L$ and $x\mid_C$ are the zero vector. Hence, the vector $x$ is the zero vector, finishing the proof.
 
\end{proof}

\subsection{Tree Decompositions}

Similarly  to the approach~\cite{tonypeter2020matroid} for the matroid secretary problem, we extensively use the tree decomposition of matroids. A tree decomposition  of bounded thickness allows us to construct non-adaptive mechanisms with good approximation ratios. Before proceeding with these constructions, let us introduce tree decompositions.

 A \emph{tree decomposition} of a matroid $M=(E,\mathcal S)$ is a pair $(T, \mathcal{X})$ where $T$ is a tree and $\mathcal{X} = \{ X_v \subseteq E\,:\, v \in V(T)\}$, where  sets in $\mathcal X$ form a partition of $E$. Here, we refer to the vertex and edge sets of the tree $T$ as $V(T)$ and $E(T)$, respectively.

 Given an edge $e = \{v_1, v_2\} \in E(T)$ of the tree $T$, let $T_1$ and $T_2$ be two connected components of  $T - e$, in other word the removal of the edge $e$ from $T$ leads to two connected components $T_1$ and $T_2$. The thickness of the edge $e=(v_1,v_2)$ is denoted as $\lambda(e)$ and is defined as follows
  \[\lambda(e) := \lambda_M(\cup_{v \in V(T_1)} X_v)\,.\]
The \emph{thickness of the tree decomposition} is the maximum thickness of the edge $e$ in $E(T)$.

Let $\mathcal{M}$ be a family of matroids, $M$ be a matroid and $(T, \mathcal{X})$ be a tree decomposition of $M$. We say that tree decomposition $(T, \mathcal{X})$ is \emph{$\mathcal{M}$-tree decomposition} if $M \mid_{cl_M(X_v)} \in \mathcal{M}$ holds for every $v\in V(T)$. Let $t_k(\mathcal{M})$ be a set of matroids which have $\mathcal{M}$-tree decomposition of thickness at most $k$.

\begin{theorem}\label{thm:td_of_representable}
  Let $\mathcal{M}_{\alpha, p}$ be the family of matroids which admit $\alpha$-competitive non-adaptive mechanisms and are representable over the finite field $\mathbb{F}_p$.
  Then for every natural number $k$ and every matroid $M$ in $t_k(\mathcal{M}_{\alpha,p})$, the matroid $M$ has an $(\alpha p^{k+1})$-competitive non-adaptive mechanism.
\end{theorem}

\begin{proof} For a natural number~$m$, let $t_{k,m}(\mathcal{M}_{\alpha, p})$ be the set of matroids which have an $\mathcal{M}_{\alpha, p}$-tree decomposition $(T, \mathcal{X})$ of thickness at most $k$ satisfying $|V(T)| = m$.

 Let us prove the statement of the lemma by induction on $m$. The base case follows from the definition of the family~$\mathcal{M}_{\alpha, p}$ and the fact that $\mathcal{M}_{\alpha, p}=t_{k,1}(\mathcal{M}_{\alpha, p})$.

Let us now show how to do the inductive step.  Let us assume $m\geq 2$ and consider a matroid~$M=(E, \mathcal S)$ in $t_{k, m}(\mathcal{M}_{\alpha, p})$ with its $\mathcal{M}_{\alpha, p}$-tree decomposition $(T, \mathcal{X})$ of thickness at most $k$ and with $|V(T)|=m$. Let $\ell$ be a leaf of the tree~$T$ and let $u$ be the neighbour of the vertex $\ell$ in the tree $T$. 

  Observe that the tree $(V(T)\setminus \{\ell\}, E(T)\setminus\{\ell u\})$ together with the subfamily $\{X_w\,:\, w\in V(T)\setminus \{\ell\}\}$ defines  an $\mathcal{M}_{\alpha, p}$-tree decomposition of the matroid $M\setminus X_\ell$. Thus the matroid~$M\setminus X_\ell$ is in $M \in t_{k, m-1}(\mathcal{M}_{\alpha, p})$. Hence, by the inductive hypothesis there are thresholds $T'_{e}$, $e\in E\setminus X_\ell$ guaranteeing $\alpha p^{k+1}$-competitiveness of the gambler in comparison to the prophet on the matroid~$M\setminus X_\ell$.

  \begin{claim} There are thresholds $T''_{e}$, $e\in X_\ell$ leading to  an $(\alpha \cdot p^{k+1})$-competitive non-adaptive mechanism for matroid $M\mid_{X_\ell}$, such that the gambler always selects a set that  is independent in $M/_{\overline{X_\ell}}$.
  \end{claim}
  \begin{proof} 
    By Lemma~\ref{lm:representable_matroids} there exists a set $S \subseteq X_l$ such that every set independent in $M\mid_S$ is also independent in the matroid $M/ \overline{X_\ell}$ and 
    \[
    \eproph_{M \mid_{S}} \ge \frac{1}{p^{k+1}} \eproph_{M{/}_{\overline{X_\ell}}}\,.
    \]
By definition of $\mathcal{M}_{\alpha, p}$ and the appearance of $X_\ell$ in the tree decomposition, we have that  $M\mid_{X_\ell}$ is in the family $\mathcal{M}_{\alpha, p}$. By Lemma~\ref{lm:apx_of_restriction}, since $S$ is a subset of $X_\ell$ the matroid~$M\mid_S$ is also in the family~$\mathcal{M}_{\alpha, p}$. Thus, there are thresholds $T''_e$, $e\in S$ that lead to an $\alpha$-competitive non-adaptive mechanism on~$M\mid_S$. The thresholds $T''_e$, $e\in X_\ell\setminus S$ can be defined as $+\infty$, finishing the proof of the claim.
  \end{proof}
    
  Now we can define thresholds~$T_e$, $e\in E$ for all elements of the matroid $M$ as follows
  \[
  T_e:=\begin{cases}
  T'_e &\text{if } e \not\in X_\ell\\
  T''_e&\text{otherwise}\,.
  \end{cases}
  \]
 Let us now demonstrate that such thresholds~$T_e$, $e\in E$ lead to an $(\alpha p^{k+1})$-competitive non-adaptive mechanism for $M$.

First, by the above claim the selected items from $X_\ell$ always form an independent set in $M/\overline{X_\ell}$ when used with the thresholds $T_e$, $e\in X_\ell$ on the matroid $M\mid_{X_\ell}$. Thus the definition of the thresholds guarantees that in expectation the value of selected items from $X_\ell$ is at least~$ \eproph_{M\mid_{X_\ell}}/(\alpha p^{k+1})$; and  in expectation the value of selected items from $E\setminus X_\ell$ is at least~$ \eproph_{M\setminus {X_\ell}}/(\alpha p^{k+1})$. 
 
To finish the proof, note that we have
 \[
 \proph_{M\mid_{X_\ell}}+ \proph_{M\setminus{X_\ell}}\geq \proph_M
 \]
 and so 
  \[
 \eproph_{M\mid_{X_\ell}}+ \eproph_{M\setminus{X_\ell}}\geq \eproph_M\,.
 \]

\end{proof}

\subsection{Regular matroids}

In this section, we prove Theorem~\ref{thm:regularcomp}.
Before we proceed to the proof, let us define key notions related to regular matroids.

 A subset of the matroid's ground set is called a \emph{circuit},  if it is an inclusion-minimal dependent set. A \emph{cycle} is a subset of the ground set which can be partitioned into a disjoint union of circuits.

 Let $M_1=(E_1, \mathcal{S}_1)$, $M_2=(E_2, \mathcal{S}_2)$ be two binary matroids. Then the matroid sum $M_1 \triangle M_2$ has the ground set
  $E_1 \triangle E_2$ and the cycles of $M_1 \triangle M_2$ are all
  sets of the form $C_1 \triangle C_2$, where $C_1$ is a cycle for $M_1$ and
  $C_2$ is a cycle for~$M_2$.

\begin{definition} Consider two binary matroids $M_1=(E_1, \mathcal{S}_1)$, $M_2=(E_2, \mathcal{S}_2)$ and $M=M_1 \triangle M_2$.

  \begin{enumerate}
  \item If $|E_1 \cap E_2| = 0$,
    and $E_1 \ne \varnothing$, $E_2 \ne \varnothing$,
    $M$ is called a \emph{$1$-sum} of $M_1$ and $M_2$.
  \item If $|E_1 \cap E_2| = 1$, $|E_1|\geq3$, $|E_2|\geq3$ and $E_1\cap E_2$ is not a loop
    of $M_1$ or $M_2$ or their dual matroids, $M$ is called a \emph{$2$-sum} of $M_1$ and $M_2$.
  \item If $|E_1 \cap E_2| = 3$, $|E_1|\geq 7$, $|E_2|\geq7$ and $E_1 \cap E_2$ is a circuit in both $M_1$ and $M_2$, and $E_1 \cap E_2$ does not contain
    a circuit in their dual matroids,
    then $M$ is called a \emph{$3$-sum} of $M_1$ and $M_2$.
  \end{enumerate}
\end{definition}

\begin{proof}[Proof of  Theorem~\ref{thm:regularcomp}] By Seymour's regular matroid decomposition theorem~\cite{seymour1980decomposition},  every regular
  matroid~$M$ can be obtained from graphic, cographic or a special matroid $R_{10}$ through a sequence of $1$-sums, $2$-sums or $3$-sums.

  This gives a tree decomposition $(T, \mathcal{X})$ of thickness at most $2$ 
  so that each~$M \mid_{X_v}$, $v\in V(T)$ is either a graphic, cographic or a special matroid ~$R_{10}$.
  
  By performing parallel extensions of the elements to be deleted before each
  $2$-sum and $3$-sum, we construct a matroid $M'$, so that $M$ is a restriction
  of $M'$ and $M'$ has a tree decomposition $(T, \mathcal{X}')$ so that each~
  $M' \mid_{\cl_{M'}(X'_v)}$, $v\in V(T)$ is either graphic, cographic or a parallel extension of~$R_{10}$.
  
  By Theorem~\ref{thm:graphic32}, every graphic matroid has a
  $32$-competitive non-adaptive mechanism. By Theorem~\ref{thm:cographic}, every cographic matroid has a $6$-competitive non-adaptive mechanism.
  Since matroid $R_{10}$ has ground set of size $10$, by Theorem~\ref{thm:twocompetitive} every parallel extension of $R_{10}$  has a $20$-competitive
  non-adaptive mechanism.

 Note that by definition every regular matroid is representable over finite field $\mathbb{F}_2$.
 Thus, by Theorem~\ref{thm:td_of_representable} with $p = 2$, $k= 2$ and $\alpha = 32$ there is  a $256$-competitive non-adaptive mechanism for matroid $M'$.
  Since $M$ is a restriction of $M'$, by Lemma~\ref{lm:apx_of_restriction},
  there is a $256$-competitive non-adaptive mechanism for $M$, finishing the proof.
\end{proof}

\subsection{Minor-closed representable matroid families}\label{sec:minor_closed}

In this section we show that every minor-clossed subclass of
matroids representable over $\mathbb{F}_p$ has a constant-competitive non-adaptive mechanism, where the constant is a function only of $p$.
The proof of this fact is analogous to the proof in~\cite{tonypeter2020matroid}.

\begin{theorem}[Theorem 4.3 in \cite{geelen2011small}]
Given natural numbers $q \ge 2$ and
  $n \ge 1$, let $M=(E,\mathcal S)$ be a matroid with no~$U_{2,q+2}$ or~$M(K_n)$ minors.
  Then we have $|E| \leq q^{q^{3n}} r_M(E)$.
\end{theorem}

\begin{corollary}\label{cor:geelen}
Given natural numbers $q \ge 2$ and
  $n \ge 1$, let $M=(E,\mathcal S)$ be a matroid with no~$U_{2,q+2}$ or~$M(K_n)$ minors.
  Then there exists a $q^{q^{3n}}$-competitive non-adaptive mechanism for~$M$.
\end{corollary}

\begin{proof} If $M$ has no~$U_{2,q+2}$ or~$M(K_n)$ minors,
  then every restriction of $M$ also has no~$U_{2,q+2}$ or~$M(K_n)$ minors.
  Thus for every $X \subseteq E$ we have $|X| \le q^{q^{3n}} r_M(X)$.
So, $M$ is  a $q^{q^{3n}}$-sparse matroid and by Theorem~\ref{thm:gamma_sparse}
 there exists a $q^{q^{3n}}$-competitive non-adaptive mechanism for~$M$.
\end{proof}

\subsubsection{Projections and lifts}

Let $M$ be a matroid and  $x$ be an element of the ground set, which  is a not a loop and not a free element of the matroid~$M$.
Then $M/x$ is called a \emph{projection} of $M\setminus x$;
$M\setminus x$ is called a \emph{lift} of $M/x$. Note that here and later we write $M/x$ and $M\setminus x$ instead of  $M/{\{x\}}$ and $M\setminus{\{x\}}$, repsectively.

Let $M$ and $N$ be two matroids with the same ground set. 
We say that the distance between $M$ and $N$ is $t$, denoted by $\dist(M, N) = t$ if $t$ is the smallest integer such that there
exists a sequence of matroids $P_0$, $P_1$, \ldots, $P_t$ where $P_0=M$ and $P_t=N$ and for every $i \in [t]$ the matroid $P_i$ is either a lift or a projection of $P_{i-1}$.

\begin{lemma}\label{lm:apx_of_lift} Let $N$ be a lift
of the matroid $M$. If there is an $\alpha$-competitive non-adaptive mechanism for $M$ then there exists a $(2\alpha+2)$-competitive non-adaptive mechanism for $N$.
\end{lemma}

\begin{proof}

Since $N$ is a lift of $M$, there exists a matroid~$L=(E,\mathcal S)$ and an element $x$ of its ground set, such that $M = L{/}x$, $N = L{\setminus}x$. Here, $x$ is not a loop and not a free element of $L$.

Let $P$ be the set of elements in $L$ that are parallel to $x$, in other words $P:= \{x' \in E\,:\,x' \parallel x\}$. Note that $N\mid_{P \setminus \{x\}}$ is a uniform matroid of rank $1$. Note also that elements in $P \setminus \{x\}$ are loops in $M$ and so $\eproph_M = \eproph_{M \setminus P}$.

Let $T'_e$, $e\in E\setminus \{x\}$ be the thresholds imposed by an $\alpha$-competitive non-adaptive mechanism for the matroid~$M$. Let $T''_e$, $e\in P$ be the thresholds guaranteeing $2$-competitive non-adaptive mechanism as in Theorem~\ref{thm:twocompetitive} for the uniform matroid of rank $1$ on the ground set  $P \setminus \{x\}$; and let $T''_e$, $e\in E\setminus (P\cup \{x\})$ be $+\infty$. 
We select one of these two sets of thresholds for the matroid $N$ as described below. The constructed mechanism for the matroid $N$
selects one of those two sets at random, where first set of thresholds $T'_e$, $e\in E\setminus \{x\}$  is selected with probability
$\gamma := \alpha{/}(\alpha+1)$ and the second set $T''_e$, $e\in E\setminus \{x\}$ with probability $1 - \gamma = 1{/}(\alpha+1)$.

Next part is dedicated to the analysis of how thresholds $T'_e$, $e\in E\setminus \{x\}$ perform on the matroid $N$. Note, that these thresholds are coming from a mechanism for the matroid $M$, while they are used for the matroid $N$ with probability $\gamma$. We show that the total expected value achieved by  thresholds $T'_e$, $e\in E\setminus \{x\}$ on $N$ is at least   the total expected value achieved by these thresholds on $M$. For this we can assume that for every realization of item values, the orders of items in matroid $N$ and $M$ are the same. To see that this assumption is valid, we can assume that the order for $N$ is chosen in an adversarial way and is used also as the items order for $M$.

\begin{claim} Let us assume that the items order for $M$ and $N$ is the same for a given realization of item values. Let us also assume that for  every item $e\in E\setminus \{x\}$ the  threshold $T'_e$ is used. Then the gambler with matroid $N$ selects all items that the gambler with matroid $M$ selects.
\end{claim}

\begin{proof} We fix the item values realization and items order. Let $e_1$, $e_2$,\ldots, $e_k$ be the items with their values being at least their threshold and with the corresponding order. 

Now we need to show that if the gambler with matroid $N$ selects items greedily from 
$e_1$, $e_2$,\ldots, $e_k$ starting from $e_1$, then the set of selected items is a superset of the items greedily selected by the gambler with matroid $M$.
 If both gamblers end up selecting exactly the same set of items, then proof of the claim is complete.  Otherwise consider the first index $i\in [k]$ such that the item $e_i$ is selected by exactly one of the two gamblers.
  Since $N = L{\setminus}x$ and $M = L{/}x$ we have that it is only possible
 if $e_i$ is selected by the gambler with the matroid $N$ and rejected by the gambler with the matroid $M$.

  Now we claim that every subsequent item, in other words an item in $e_{i+1}$, \ldots, $e_k$, is either selected by both gamblers or rejected by both gamblers. 
  Suppose the contrary and consider the first item $e_j$, $i+1\leq j\leq k$ that is selected by one gambler and rejected by another gambler. Let $S:= \{e_1, e_2, \ldots, e_{j-1}\}$ and let $T$ be the set of items selected by the gambler with $M$ from the set $S$. Thus the gambler with $N$ selected $T \cup \{e_i\}$ from the set $S$. So $T \cup \{e_i\}$ is a basis of $(L\setminus x)\mid_S$ and $T$ is a basis of $(L/x)\mid_S$.
Thus, both $T \cup \{e_i\}$   and $T \cup \{x\}$ are bases of   $L\mid_S$. If only one of the two gamblers accepts the item $s_j$ then the matroid  $L\mid_{S\cup \{s_j\}}$ has two bases of different cardinality, attaining a contradiction and finishing the proof.
 \end{proof}
 \medskip

Thus we have that the thresholds $T'_e$, $e\in E\setminus \{x\}$ guarantee at least $\eproph_M$ as the expected total value of the gambler with $N$. To prove that the constructed mechanism is $1/(2\alpha+2)$-competitive it is enough to show the following claim. Note that in our construction we used $\alpha$-competitive non-adaptive mechanism for the matroid $M$ and $2$-competitive non-adaptive mechanism for the uniform matroid of rank $1$ on $P\setminus\{x\}$.

\begin{claim} $\gamma\frac{1}{\alpha}\eproph_M+ (1-\gamma) \frac{1}{2}\eproph_{P\setminus \{x\}} \ge \frac{1}{2\alpha+2} \eproph_N$
\end{claim}

\begin{proof} 
  
Let us consider the inclusion-maximal set $I_{opt}$ on which the prophet achieves $\proph_N$. Let $C_{opt}$ be a random variable corresponding to the unique circuit of $I_{opt} \cup \{x\}$ in $L$. Recall that $x$ is not a free element of $L$ so such a circuit exists and is unique and contains $x$.

First consider the events when $|C_{opt}| \ge 3$.
Note that by definition of a circuit, for every $y \in C_{opt}\setminus \{x\}$ the set $(I_{opt}\cup \{x\}) \setminus \{y\}$ is independent in $L$. 
Hence, for every $y \in C_{opt}\setminus \{x\}$ the set  $I_{opt} \setminus \{y\}$ is independent in $M$. So we have that
conditioned on $|C_{opt}| \ge 3$ we have $\proph_M \geq w(I_{opt} \setminus \{y\})$ for every $y \in C_{opt}\setminus \{x\}$.
Let $y_{opt}$ be the random variable representing the element in $C_{opt}\setminus \{x\}$ of smallest value. Then conditioned on $|C_{opt}| \ge 3$, we have $w(C_{opt} \setminus \{y_{opt},x\}) \geq  w(C \setminus \{x\})/2$. Thus, conditioned on  $|C_{opt}| \ge 3$ we have
\begin{align*}
\proph_M \ge w(I_{opt} \setminus \{y_{opt}\}) = w(I_{opt}  \setminus C_{opt}) + w(C_{opt} \setminus \{y_{opt}\})\\\geq w(I_{opt} \setminus C_{opt}) + \frac{1}{2} w(C_{opt}\setminus \{x\}) \geq \frac{1}{2}  w(I_{opt}) = \frac{1}{2} \proph_N\,.
\end{align*}

Second consider the event that $|C_{opt}| < 3$. Since $x$ is not a loop of $L$ by definition, we have $|C_{opt}| = 2$ and
so $C_{opt} = \{x, x_{opt}\}$ for some random variable element $x_{opt} \in P \setminus \{x\}$. For the event $|C_{opt}| \geq 3$ let us define the random variable element $x_{opt}$ to be an arbitrary element in $C_{opt}\setminus \{x\}$. Thus,  if $|C_{opt}| < 3$ we have $\proph_{P\setminus\{x\}} \geq  w(x_{opt})$.
Now let us define $J_{opt} := I_{opt} \setminus \{x_{opt}\}$ and note that $J_{opt}$ is independent in the matroid~$M$. Moreover, since $I_{opt}$ is the set on which the prophet achieves $\proph_N$, 
we have that conditioned on $|C_{opt}| < 3$ the prophet achieves $\proph_M$ on the set ~$J_{opt}$.

Combining everything together we have
\begingroup
\allowdisplaybreaks
\begin{align*}
&\gamma\frac{1}{\alpha}\eproph_M+ (1-\gamma) \frac{1}{2}\eproph_{P\setminus \{x\}}=\\
&\frac{1}{\alpha+1}\eproph_M+ \frac{1}{2\alpha+2}\eproph_{P\setminus \{x\}}\geq\\
&\qquad E \left[\frac{w(x_{opt})}{2\alpha+2} + \frac{\proph_M}{\alpha+1} \,\middle\vert\, |C_{opt}| <3 \right] P\left[|C_{opt}| <3\right] \\
&\qquad+E\left[\frac{\proph_M}{\alpha+1} \,\middle\vert\, |C_{opt}| \ge 3 \right] P\left[|C_{opt}| \ge 3\right] =\\
&\qquad E \left[\frac{w(x_{opt})}{2\alpha+2} + \frac{w(I_{opt}\setminus \{x_{opt}\})}{\alpha+1} \,\middle\vert\, |C_{opt}| <3 \right] P\left[|C_{opt}| <3\right] \\
&\qquad+E\left[\frac{\proph_M}{\alpha+1} \,\middle\vert\, |C_{opt}| \ge 3 \right] P\left[|C_{opt}| \ge 3\right] \geq\\ 
&\qquad E \left[\frac{\proph_N}{2\alpha+2} \,\middle\vert\, |C_{opt}| <3 \right] P\left[|C_{opt}| <3\right]\\
&\qquad+E\left[\frac{\proph_M}{\alpha+1} \,\middle\vert\, |C_{opt}| \ge 3 \right] P\left[|C_{opt}| \ge 3\right]\geq\frac{1}{2\alpha+2}\eproph_N\,.
\end{align*}
\endgroup
\end{proof}

\end{proof}

\begin{lemma}\label{lem:sequence_projections} Let $N$ be a matroid obtained from a matroid $M$ by a sequence of $t$ projections. Let $L$ be the set of loops in the matroid~$N$. Let there exist an $\alpha$-competitive non-adaptive mechanism for the matroid~$M$.   Then there exists a non-adaptive mechanism for $N{\setminus}L$ such that the expected total value of this mechanism is at least~$\frac{1}{\alpha \cdot 3^t} \eproph_{M{\setminus}L}$.
\end{lemma}

In the context of Lemma~\ref{lem:sequence_projections}, every set that is independent for the matroid $N\setminus L$ is also independent for the matroid $M\setminus L$. Hence, we have~$\eproph_{M\setminus L}\geq \eproph_{N\setminus L}$. Thus in case $t=1$, Lemma~\ref{lem:sequence_projections}  leads us to the following corollary.

\begin{corollary}\label{cor:apx_with_projection}

Let $N$ be a projection
of the matroid $M$. If there is an $\alpha$-competitive non-adaptive mechanism for $M$ then there exists a $3\alpha$-competitive non-adaptive mechanism for $N$.
\end{corollary}

\begin{proof}[Proof of Lemma~\ref{lem:sequence_projections}] Let us prove the statement by induction. Of course, in case $t=0$ we have $M=N$ and the statement is trivially true. 

Let us now assume that $t$ is at least $1$. Let $N'$ be a matroid such that $N'$ is obtained from the matroid~$M$ by a sequence of $t-1$ projections and $N$ is a projection of $N'$. Since $N$ is a projection of $N'$ there is a matroid $P=(E,\mathcal S)$ and $x\in E$ such that $P\setminus x =N'$ and $P/x=N$. Let $L'$ be the set of loops in the matroid~$N'$.

  By induction hypothesis, there exist thresholds $T'_e$, $e\in E\setminus (L'\cup \{x\})$ such that the gambler with the matroid~$N'{\setminus}L'$ achieves at least $\frac{1}{\alpha \cdot 3^{t-1}} \eproph_{M{\setminus}L'}$ as the expected total value.  Let us assume that to  compute thresholds $T'_e$, $e\in E\setminus (L'\cup \{x\})$ the values of items in $L$ were set to be~$0$ while the distribution of values for other items remain the same.
Since $L' \subseteq L$,  analogously to Lemma~\ref{lm:apx_of_restriction} we can define thresholds 
\[
T''_e:=
\begin{cases} 
+\infty&\text{if  } e\in L\\
T'_e&\text{otherwise}
\end{cases}
\]
such that the gambler with the matroid~$N'{\setminus}L$ achieves at least $\frac{1}{\alpha \cdot 3^{t-1}} \eproph_{M{\setminus}L}$ as the expected total value.  Let $T'''_e$, $e\in E\setminus (L\cup \{x\})$ be the thresholds guaranteeing $2$-competitive non-adaptive mechanism as in Theorem~\ref{thm:twocompetitive} for the uniform matroid of rank $1$ on the ground set  $E\setminus (L\cup \{x\})$. 

The constructed mechanism for the matroid $N\setminus L$
selects one of two threshohold sets at random, where first set of thresholds $T''_e$, $e\in E\setminus (L\cup \{x\})$  is selected with probability
$1/3$ and the thresholds $T'''_e$, $e\in E\setminus (L\cup \{x\})$ with probability $2/3$. Note that the thresholds $T''_e$, $e\in E\setminus (L\cup \{x\})$ were designed for the matroid $N'\setminus L$ but are used for the matroid $N\setminus L$; hence less items might be selected than when it is used for $N'\setminus L$. Also note, that  the thresholds $T'''_e$, $e\in E\setminus (L\cup \{x\})$ are used for $N\setminus L$ but were designed for the uniform matroid of rank $1$. 

For the analysis, let $I_{alg}$ be the random variable indicating the items set selected by the gambler with matroid $N'\setminus L$ when the thresholds $T''_e$, $e\in E\setminus (L\cup \{x\})$ are used. Analogously to a claim in the proof of Lemma~\ref{lm:apx_of_lift}, we can assume that when the thresholds $T''_e$, $e\in E\setminus (L\cup \{x\})$ are used the gambler with $N\setminus L$ select all items in $I_{alg}$ with an exception for possibly one item. Let $x_{opt}$ be the random variable indicating the element of maximum value in $E\setminus (L\cup \{x\})$. 

To finish the proof it is enough to show the following inequality
\[
\frac{1}{3} E[w(I_{alg})-w(x_{opt})]+\frac{2}{3}\frac{1}{2}E[w(x_{opt})]\geq \frac{1}{\alpha \cdot 3^t} \eproph_{M{\setminus}L}\,.
\]
To obtain this inequality we can do estimations as follows
\[
\frac{1}{3} E[w(I_{alg})-w(x_{opt})]+\frac{2}{3}\frac{1}{2}E[w(x_{opt})]=\frac{1}{3}E[w(I_{alg})]\geq\frac{1}{3}\frac{1}{\alpha \cdot 3^{t-1}} \eproph_{M{\setminus}L}\,.
\]
\end{proof}

Now let us combine Corollary~\ref{cor:apx_with_projection} and Lemma~\ref{lm:apx_of_lift}.

\begin{lemma}\label{lm:dist_apx}
Let $M$ and $N$ be matroids such that $\dist(M, N) \leq t$. If there exists an $\alpha$-competitive non-adaptive mechanism for the matroid $M$ with $\alpha\geq 2$ then  there exists a $3^t \alpha$-competitive non-adaptive mechanism for the matroid $N$.
\end{lemma}
\begin{proof}
Note that for $\alpha\geq 2$ we have $3\alpha\geq 2\alpha +2$. Since $N$ can be obtained from $M$ by a sequence of $t$ projection and lift steps, we can use
Corollary~\ref{cor:apx_with_projection} or Lemma~\ref{lm:apx_of_lift} for each of these steps to obtain the desired competitiveness ratio. 
\end{proof}

\subsubsection{Minor-closed families theorem}

\begin{lemma}[Lemma 6 in \cite{tonypeter2020matroid}]\label{lm:u2p2_not_fp}
Let  $p$ and $n$ be integers such that $p \le n - 2$ and $p$ is prime. The matroid $U_{2, n}$ is not representable over the field $\mathbb{F}_p$.
\end{lemma}

The following Structural Hypothesis is due to Geelen, Gerards and Whittle. The proof of this Structural Hypothesis has not appeared in print.

\begin{hypothesis}\label{hypothesis:tree_decomp} Let $p$ be a prime number and $\mathcal{M}$ is a proper minor-closed  class of matroids representable over $\mathbb{F}_p$.

  Then there exist $k$, $n$, $t$ such that every $M \in \mathcal{M}$ is a restriction
  of an $\mathbb{F}_p$-representable matroid $M'$ having a full tree-decomposition $(T, \mathcal X)$ of
  thickness at most $k$ so that for every $v \in V(T)$ if $M' \mid_{\cl_{M'}(X_v)}$
  has a $M(K_n)$ minor, then there exists a $2$-column sparse matroid $N$ with
  $\dist(M' \mid_{\cl_{M'}(X_v)}, N) \le t$.
\end{hypothesis}

\begin{proof}[Proof of Theorem~\ref{thm:minor_closed}]  Let $k$, $n$, $t$ are as stated in the Structural Hypothesis~\ref{hypothesis:tree_decomp} on $\mathcal{M}$.

  Let $\mathcal{M}_1$ be the set of matroids on distance $t$ or less from some $2$-column sparse matroid and are representable over $\mathbb{F}_p$.
  By Theorem~\ref{thm:k_column_sparse} all $2$-column
  sparse matroids have a $32$-competitive non-adaptive mechanism. By Lemma~\ref{lm:dist_apx}
  there exists a $(3^t \cdot 32)$-competitive mechanism for matroids in $\mathcal{M}_1$.
  
  Let $\mathcal{M}_2$ be the set of matroids without $M(K_n)$ minor and are
  representable over $\mathbb{F}_p$. By Lemma~\ref{lm:u2p2_not_fp} all matroids
  in $\mathcal{M}_2$ do not have $U_{2,p+2}$ as a minor.
  Then by Corollary~\ref{cor:geelen}, we have that there is a
  $p^{p^{3n}}$-competitive
  non-adaptive mechanism for every matroid in $\mathcal{M}_2$.

By the Structural Hypothesis~\ref{hypothesis:tree_decomp} we have that every $M \in \mathcal{M}$
  is a restriction of some $M'$ with a full tree-decomposition $(T, \mathcal X)$ of
  thickness at most $k$ so that for every $v \in V(T)$ $M' \mid_{\cl_{M'}(X_v)} \in \mathcal{M}_1 \cup \mathcal{M}_2$.

  Thus by Theorem~\ref{thm:td_of_representable}, matroid $M'$ has a $\gamma := (\max(3^t \cdot 32, p^{p^{3n}}) \cdot p^{k+1})$-competitive non-adaptive mechanism. By Lemma~\ref{lm:apx_of_restriction} the
  matroid $M$ has also a $\gamma$-competitive non-adaptive mechanism.

\end{proof}

\newpage
\bibliographystyle{alpha}
\bibliography{literature.bib}

\begin{thebibliography}{CGKM20}

\bibitem[AKW19]{azar}
Pablo~D. Azar, Robert Kleinberg, and S.~Matthew Weinberg.
\newblock Prior independent mechanisms via prophet inequalities with limited
  information.
\newblock {\em Games and Economic Behavior}, 118:511--532, 2019.

\bibitem[CFPP21]{caramanis}
Constantine Caramanis, Matthew Faw, Orestis Papadigenopoulos, and Emmanouil
  Pountourakis.
\newblock Single-sample prophet inequalities revisited.
\newblock {\em ArXiv}, abs/2103.13089, 2021.

\bibitem[CGKM20]{chawla2020non}
Shuchi Chawla, Kira Goldner, Anna~R Karlin, and J~Benjamin Miller.
\newblock Non-adaptive matroid prophet inequalities.
\newblock {\em arXiv preprint arXiv:2011.09406}, 2020.

\bibitem[CHMS10]{chawla2010multi}
Shuchi Chawla, Jason~D Hartline, David~L Malec, and Balasubramanian Sivan.
\newblock Multi-parameter mechanism design and sequential posted pricing.
\newblock In {\em Proceedings of the forty-second ACM symposium on Theory of
  computing}, pages 311--320, 2010.

\bibitem[DK14]{dinitz}
Michael Dinitz and Guy Kortsarz.
\newblock Matroid secretary for regular and decomposable matroids.
\newblock {\em SIAM Journal on Computing}, 43(5):1807--1830, 2014.

\bibitem[DK15]{dutting2015polymatroid}
Paul D{\"u}tting and Robert Kleinberg.
\newblock Polymatroid prophet inequalities.
\newblock In {\em Algorithms-ESA 2015}, pages 437--449. Springer, 2015.

\bibitem[Edm65]{edmonds1965minimum}
Jack Edmonds.
\newblock Minimum partition of a matroid into independent subsets.
\newblock {\em J. Res. Nat. Bur. Standards Sect. B}, 69:67--72, 1965.

\bibitem[FSZ16]{feldman2016online}
Moran Feldman, Ola Svensson, and Rico Zenklusen.
\newblock Online contention resolution schemes.
\newblock In {\em Proceedings of the twenty-seventh annual ACM-SIAM symposium
  on Discrete algorithms}, pages 1014--1033. SIAM, 2016.

\bibitem[FSZ21]{feldman2021online}
Moran Feldman, Ola Svensson, and Rico Zenklusen.
\newblock Online contention resolution schemes with applications to bayesian
  selection problems.
\newblock {\em SIAM Journal on Computing}, 50(2):255--300, 2021.

\bibitem[Gee11]{geelen2011small}
Jim Geelen.
\newblock Small cocircuits in matroids.
\newblock {\em European Journal of Combinatorics}, 32(6):795--801, 2011.

\bibitem[GW19]{gravin2019prophet}
Nikolai Gravin and Hongao Wang.
\newblock Prophet inequality for bipartite matching: Merits of being simple and
  non adaptive.
\newblock In {\em Proceedings of the 2019 ACM Conference on Economics and
  Computation}, pages 93--109, 2019.

\bibitem[HN20]{tonypeter2020matroid}
Tony Huynh and Peter Nelson.
\newblock The matroid secretary problem for minor-closed classes and random
  matroids.
\newblock {\em SIAM Journal on Discrete Mathematics}, 34(1):163--176, 2020.

\bibitem[KS77]{krengel1977semiamarts}
Ulrich Krengel and Louis Sucheston.
\newblock Semiamarts and finite values.
\newblock {\em Bulletin of the American Mathematical Society}, 83(4):745--747,
  1977.

\bibitem[KW12]{kw2012matroid}
Robert Kleinberg and Seth~Matthew Weinberg.
\newblock Matroid prophet inequalities.
\newblock In {\em Proceedings of the forty-fourth annual ACM symposium on
  Theory of computing}, pages 123--136, 2012.

\bibitem[MTW16]{ma}
Tengyu Ma, Bo~Tang, and Yajun Wang.
\newblock The simulated greedy algorithm for several submodular matroid
  secretary problems.
\newblock {\em Theory of Computing Systems}, 58(4):681--706, 2016.

\bibitem[Mye81]{myerson1981optimal}
Roger~B Myerson.
\newblock Optimal auction design.
\newblock {\em Mathematics of operations research}, 6(1):58--73, 1981.

\bibitem[Oxl06]{oxley}
James~G. Oxley.
\newblock {\em Matroid Theory}.
\newblock Oxford graduate texts in mathematics. Oxford University Press, 2006.

\bibitem[SC84]{samuelcahn1984comparison}
Ester Samuel-Cahn.
\newblock Comparison of threshold stop rules and maximum for independent
  nonnegative random variables.
\newblock {\em the Annals of Probability}, pages 1213--1216, 1984.

\bibitem[Sch03]{schrijver2003combinatorial}
A.~Schrijver.
\newblock {\em Combinatorial Optimization: Polyhedra and Efficiency}.
\newblock Number Bd. 1 in Algorithms and Combinatorics. Springer, 2003.

\bibitem[Sey80]{seymour1980decomposition}
Paul~D Seymour.
\newblock Decomposition of regular matroids.
\newblock {\em Journal of combinatorial theory, Series B}, 28(3):305--359,
  1980.

\bibitem[Sot13]{soto2013matroid}
Jos{\'e}~A Soto.
\newblock Matroid secretary problem in the random-assignment model.
\newblock {\em SIAM Journal on Computing}, 42(1):178--211, 2013.

\end{thebibliography}

\end{document}